\documentclass[
    reprint,
    aps,
    prx,
    twocolumn,
    footinbib
]{revtex4-2}

\usepackage[
    colorlinks=true,
    linkcolor=blue,
    citecolor=red,
    plainpages=false
]{hyperref}

\usepackage{graphicx}
\usepackage{amsfonts,amssymb, amsmath, amsthm}
\usepackage{bm}
\usepackage{xcolor}
\usepackage{mathtools}
\usepackage{physics}
\usepackage{algorithm}
\usepackage[noend]{algpseudocode}
\usepackage{enumitem}
\usepackage{tabularray}
\UseTblrLibrary{booktabs}
\usepackage{tikz}
\usepackage{array,booktabs,makecell}
\newcolumntype{C}{>{\raggedright\arraybackslash}X}
\SetTblrStyle{note}{halign=l}      

\newtheorem{theorem}{Theorem}
\newtheorem{lemma}[theorem]{Lemma}
\newtheorem{claim}[theorem]{Claim}
\newtheorem{corollary}[theorem]{Corollary}

\DeclareMathOperator{\poly}{poly}
\DeclareMathOperator{\polylog}{polylog}
\DeclareMathOperator{\vecop}{vec}
\let\abs\relax  
\DeclarePairedDelimiter{\abs}{\vert}{\rvert}
\DeclarePairedDelimiter{\opnorm}{\lVert}{\rVert}
\DeclarePairedDelimiter{\ceil}{\lceil}{\rceil}
\DeclarePairedDelimiterXPP\tracenorm[1]{}{\lVert}{\rVert}{_1}{#1}
\DeclarePairedDelimiterXPP\diamondnorm[1]{}{\lVert}{\rVert}{_\diamond}{#1}
\DeclarePairedDelimiterXPP\lnorm[1]{}{\lVert}{\rVert}{_2}{#1}
\DeclarePairedDelimiterX\innerp[3]{\langle}{\rangle}{#1\,\delimsize\vert\,\mathopen{}#2\,\delimsize\vert\,\mathopen{}#3}%

\providecommand\given{}
\newcommand\SetSymbol[1][]{%
\nonscript\:#1\vert
\allowbreak
\nonscript\:
\mathopen{}}
\DeclarePairedDelimiterX\Set[1]\{\}{%
\renewcommand\given{\SetSymbol[\delimsize]}
#1
}

\newcommand{\eg}{e.g.,\@}
\newcommand{\ie}{i.e.,\@}

\begin{document}

\title{Probabilistic quantum algorithm for Lyapunov equations and matrix inversion}

\author{Marcello Benedetti}
\email{marcello.benedetti@quantinuum.com}
\affiliation{Quantinuum, Partnership House, Carlisle Place, London SW1P 1BX, United Kingdom} 

\author{Ansis Rosmanis}
\email{ansis.rosmanis@quantinuum.com}
\affiliation{Quantinuum, Partnership House, Carlisle Place, London SW1P 1BX, United Kingdom} 

\author{Matthias Rosenkranz}
\email{matthias.rosenkranz@quantinuum.com}
\affiliation{Quantinuum, Partnership House, Carlisle Place, London SW1P 1BX, United Kingdom} 

\date{April 17, 2026}

\begin{abstract}
We present a probabilistic quantum algorithm for preparing mixed states
which, in expectation, are proportional to the solutions of Lyapunov
equations---linear matrix equations ubiquitous in the analysis of classical
and quantum dynamical systems. Building on previous results by Zhang \textit{et al.}, 
\href{https://arxiv.org/abs/2304.04526}{arXiv:2304.04526}, at each step the
algorithm can (i) return the current state, (ii) apply a trace nonincreasing completely positive map, or (iii) restart. We introduce a deterministic stopping rule, which leads to an efficient algorithm with a bounded expected number of calls to oracles representing the two input matrices of the Lyapunov equations. We also consider preparing a mixed state that approximates the normalized inverse of a
positive definite matrix $A$. In its most general form, the algorithm generates mixed states, which approximate matrix-valued weighted sums and integrals. It can be shown that block encodings and states yield two incomparable computational resources even when they represent the same piece of data. While block encodings of functions have received much attention in the literature, our work takes a step toward the less explored problem of encoding functions into mixed states.
\end{abstract}

\maketitle

\section{Introduction}

We present a probabilistic quantum
algorithm that solves several linear matrix equations under certain conditions in the following sense. Given an $N\times N$ normal matrix $A$ and $N\times N$ positive semidefinite matrix $B$, in expectation the algorithm generates mixed states approximately proportional to solutions $X$ of the
\emph{Lyapunov equations}
\begin{align}
    AXA^\dagger - X + B
    &=
    0
    &&\text{(discrete time)}\label{eq:lyapunov_discrete}\\
    AX + XA^\dagger + B
    &=
    0
    &&\text{(continuous time)}\label{eq:lyapunov_continuous}
    \intertext{with bounded expected stopping time and guarantees on the approximation error. We also discuss an application of
    the algorithm to finding solutions of}
    AX &= I
    &&\text{(matrix inversion)}\label{eq:matrix_inversion}
\end{align}
with error $\epsilon$ in trace distance and expected stopping time $\ceil{ \kappa\ln(1/\epsilon)}+1$ for positive definite $A$ with condition number $\kappa$ and $I$ the identity.

Lyapunov equations are foundational to control and systems theory, fields in engineering concerned with controlling and predicting dynamical systems (see, \eg\ Refs.~\cite{antoulas2005approximation,gajicLyapunovMatrixEquation2008} and references therein). Consider a continuous-time, linear dynamical system $\dot{\bm{x}} = A\bm{x}$ with $\bm{x}$ the system state. More general dynamical systems include driving terms and outputs. If this system is stable [here, $\bm{x}(t)\rightarrow 0$ for $t\rightarrow\infty$], a solution to the Lyapunov equation exists for positive semidefinite $B$. Then functions of the form $\bm{x}^T X \bm{x}$ characterize energy transfer and the approach to equilibrium. Lyapunov equations also determine the time evolution of covariance matrices of linear open quantum systems~\cite{bravyiClassicalSimulationDissipative2011,purkayasthaLyapunovEquationOpen2022}. Furthermore, their solutions appear in the calculation of the Fr{\'e}chet derivatives of the matrix square root function~\cite{Delmoral_2018}, and in quantum information as the symmetric logarithmic derivative used in the quantum Fisher information matrix and Bures metric~\cite{Paris_2009, wojtowiczQuantumFisherInformation2025a}.

A conceptually simple way to solve either of the Lyapunov equations is to map it to a
system of $N^2$ linear equations $\mathcal{A}\bm{x} = \bm{b}$ and find an
approximate solution $\tilde{\bm{x}}$ with $\lnorm{\tilde{\bm{x}} - \bm{x}}
< \epsilon$. For general or sparse $\mathcal{A}$, this is impractical for
large $N$ owing to a $\poly(N)$ runtime of any classical solver. Specialized solvers, such as the
Bartels-Stewart algorithm~\cite{bartelsAlgorithm432C21972}, or convex optimization~\cite{boydLinearMatrixInequalities1994} polynomially reduce the complexity but remain restricted to moderate $N$. Approximate methods for larger-scale systems include singular value decomposition-based methods and Krylov methods~\cite{antoulas2005approximation,simonciniComputationalMethodsLinear2016}. Alternatively, one can use a quantum linear system solver to
obtain a state $\ket{\tilde{\bm{x}}}$ with $\lnorm{\ket{\tilde{\bm{x}}} -
\ket{\bm{x}}} < \epsilon$ for $\ket{\bm{x}} \propto \sum_k x_k
\ket{k}$~\cite{sunSolvingLyapunovEquation2017}. Given an oracle that block encodes $\mathcal{A}$ with $\opnorm{\mathcal{A}} = 1$ and condition number $\kappa$, and an oracle preparing $\ket{\bm{b}}$, the best currently known quantum linear system solvers achieve query complexity $\mathcal{O}[\kappa \log(1/\epsilon)]$~\cite{costaOptimalScalingQuantum2022,dalzellShortcutOptimalQuantum2024}. Thus, if one can construct the block encoding and state-preparation circuits with $\polylog(N)$ gates, preparing the state $\ket{\tilde{\bm{x}}}$ has runtime $\mathcal{O}[\kappa\log(1/\epsilon)\polylog(N)]$. This approach has the potential for a large scaling advantage in $N$ in applications that only require estimating few expectation values rather than reading out the full state $\tilde{\bm{x}}$. For a recent
review of quantum linear system solvers, see Ref.~\cite{Morales_2025}. A quantum algorithm for the continuous-time Lyapunov equation was proposed in Ref.~\cite{claytonDifferentiableQuantumComputing2024}. This equation is also a special case of the algebraic Riccati equation, for which a quantum algorithm is given in Ref.~\cite{Liu_2025}, and of the Sylvester equation, for which a quantum algorithm appeared very recently~\cite{Somma_2025}. These approaches output a block encoding of the subnormalized solution matrix.

\begin{table*}[tb]
    \centering
    \begin{tblr}{colspec={lllll}}
        \toprule
        \SetCell[r=2]{h} \textbf{Problem}
        & \SetCell[r=2]{h} \textbf{Assumption}
        & \SetCell[c=3]{c} \textbf{Parameters}
        &
        &
        \\
        \cmidrule{3-5}
        &
        & $T$
        & $M$ [for $\mathcal{E}(\cdot)=M\,\cdot \, M^\dagger$]
        & $\rho_0$
        \\
        \midrule
        Discrete-time Lyapunov
        & $\opnorm{A} < 1$
        & $\ceil*{\frac{\ln(1/\epsilon)}{2\ln(1/\opnorm{A})}}$
        & $A$
        & $B$
        \\
        Continuous-time Lyapunov
        & $\Re[\lambda(A)] < 0$
        & $\ceil*{ \frac{\ln(1/\epsilon_2)}{2\epsilon_1} \frac{ \opnorm{A}\abs*{\min \Set{ \Re[\lambda(A)] }}}{ \abs*{\max \Set{ \Re[\lambda(A)] }}^2 } }$
        & $e^{\frac{\epsilon_1}{\opnorm{A}}\frac{\max\Set{\Re[\lambda(A)]}}{\min\Set{\Re[\lambda(A)]}} A}$
        & $B$
        \\
        Matrix inversion I
        & $A \succ 0$
        & $\ceil*{ \kappa \ln (\frac{1}{ \epsilon}) }$
        & $\sqrt{I -\frac{A}{\opnorm{A}}}$
        & $\frac{I}{N}$
        \\
        Matrix inversion II
        & $A \succ 0$
        & $\ceil*{ \frac{\kappa^2}{\epsilon} \ln(\frac{2}{\epsilon}) }$
        & $e^{-\frac{\epsilon}{2\kappa \opnorm{A}} A}$
        & $\frac{I}{N}$
        \\
        \bottomrule
    \end{tblr}
    \caption{\textbf{Overview of results.} Algorithm~\ref{alg:cap} achieves trace distance $\epsilon$ or $\epsilon_1+\epsilon_2$ between the expected output state and normalized solutions $X/\tr(X)$ of the listed problems using the stated parameters. We set parameter $r_k = \frac{1}{T+1-k}$ for all $k \in [T+1]$ and assume $B\succeq 0$, $\tr(B)=1$, $[A, A^\dagger] = 0$. On average the algorithm queries a block encoding of $M$ and a state-preparation circuit for $\rho_0$ at most $T+1$ times.
    $\lambda(A)$ is the set the eigenvalues of $A$ and $\kappa$ the condition number of $A$.}
    \label{tab:summary}
\end{table*}

In contrast to previous work, here we want to produce quantum states proportional to the solution of the matrix equation, \ie\ $\rho \propto X$. This is interesting because, in general, a state and a block encoding for the same piece of data provide incomparable computational resources. We prove this in Appendix~\ref{app:incompatibility} by adapting a result from Ref.~\cite{Somma_2025}. Our algorithm builds on the work by~\citet{zhang2023dissipative}. Their
algorithm generates states via a stochastic process $\Set{\rho_t}$, which depends on the outcomes of both a classical coin flip and a quantum measurement of an ancilla register. Coin flips determine whether the process stops, while measurements determine whether the desired completely positive (CP) map $\mathcal{E}$ acting on the main register is correctly implemented. The stopped process generates states
proportional to $\sum_{k=0}^\infty c_k \mathcal{E}^k(\frac{I}{N})$ in expectation,
where $c_k$ depends on the coin-flip probabilities $\Set{r_j}_{j=0}^k$ for all $k$ and $\mathcal{E}^k$ means applying $\mathcal{E}$ $k$ times. For a
suitable choice of $\mathcal{E}$ and $\Set{r_k}$, this process returns an approximate
quantum Gibbs state in expectation. 

We modify the algorithm to generate states proportional to
$\sum_{k=0}^T c_k \mathcal{E}^k(\rho_0)$ with an arbitrary initial state
$\rho_0$. The main difference to Ref.~\cite{zhang2023dissipative} is the newly introduced parameter $T$ and different choices of $\mathcal{E}$, $\Set{c_k}$, and $\rho_0$. These modifications allow us to
generate states satisfying
$\frac{1}{2}\tracenorm[\Big]{\mathbb{E}[\rho_{}] - \frac{X}{\tr(X)}} \leq
\epsilon$ for solutions $X$ of either the Lyapunov equations or matrix
inversion. For Lyapunov equations we assume $\tr(B)=1$ and show that, in expectation, our
algorithm requires at most $T+1$ calls to a block encoding of either $A$ or $e^{\Delta A}$ for small $\Delta$ and
a state-preparation circuit for $\rho_0=B$. The same result holds for matrix inversion where we use $\rho_0=I/N$ instead. Once the state is prepared, it can be used for the estimation of expectations values $\frac{\tr(X O)}{\tr(X)}$, overlaps $\frac{\innerp{\psi}{X}{\phi}}{\tr(X)}$, or the normalization $\tr(X)$. Table~\ref{tab:summary} summarizes the main results. Approximate implementations of the map $\mathcal{E}$ lead to an additive error discussed in the main text.

For matrix inversion our algorithm appears to match the optimal query complexity bound in $\kappa, \epsilon$ for quantum linear system solvers~\cite{costaDiscreteAdiabaticQuantum2025,orsucciSolvingClassesPositivedefinite2021,mori2026sparsity}. However, it is worthwhile to emphasize the subtle difference: for quantum linear systems we want to output a pure state $\ket{\bm{x}} \propto A^{-1} \ket{\bm{b}}$, for matrix inversion we want to output a mixed state $\rho_{} \propto A^{-1}$. While the former problem is known to be BQP complete~\cite{Harrow_2009}, we are not aware of the complexity class of the latter. Finally, we compare quantum algorithms for solving the continuous-time Lyapunov equation in Appendix~\ref{app:comparison_of_algos}. Bearing in mind the distinction in output formats (block encodings versus states), the query complexity analysis indicates that our algorithm remains competitive while relying on potentially simpler primitives.

\section{Background}

We begin with a description of the probabilistic quantum algorithm by~\citet{zhang2023dissipative}. Throughout this paper $\tracenorm{\cdot}$
denotes the trace norm, $\opnorm{\cdot}$ the operator norm, $\lambda(A)$ the set of
eigenvalues of $A$, and for integer $i>0$, $[i]\coloneq \Set{0, 1, \dots, i-1}$.

The algorithm starts with an initial quantum state $\rho_0$ and iteration number
$i=0$. At each iteration $i$ there are up to two probabilistic events. The first
event is given by a Bernoulli random variable $b \sim \mathcal{B}(r_i)$ where $0
\leq r_i \leq 1$ is the probability of observing the $b = 1$ outcome. When $b =
1$ is observed, the algorithm stops and returns the current state $\rho_i$. When
$b = 0$, the algorithm applies a channel to the current state:
\begin{align}
\label{eq:cptp_map_1}
    \Phi(\rho_i)
    =
    \dyad{0} \otimes \mathcal{F}(\rho_i) + \dyad{1} \otimes \mathcal{E} (\rho_i),
\end{align}
where $\mathcal{E},
\mathcal{F}$ are two trace nonincreasing CP maps that form a quantum
instrument. 
The channel $\Phi$ can be interpreted as either the production of a desired state
$\frac{\mathcal{E}(\rho_i)}{\tr[\mathcal{E}(\rho_i)]}$ with probability
$p_{i+1}=\tr[\mathcal{E}(\rho_i)]$, or the production of an undesired state
$\frac{\mathcal{F}(\rho_i)}{1 - \tr[\mathcal{E}(\rho_i)]}$ with probability $1 - p_{i+1}$. Therefore, a
second probabilistic event is triggered by measuring the ancilla qubit in
the computational basis. This is another Bernoulli random variable $a \sim \mathcal{B}(p_{i+1})$. A measurement outcome $a=1$ indicates that the desired
state has been produced. In this case we
increment the iteration number to $i+1$ and proceed with a recursive call of the
iteration just described. If $a = 0$, we introduce a fresh copy of the initial
state $\rho_0$ and restart the algorithm at iteration $i=0$.

We briefly mention a model for an exact implementation of the CP maps and discuss its approximate implementation later. For our purposes it is
sufficient to focus on the map $\mathcal{E}$ and assume a single, square Kraus operator $M$ such that 
\begin{equation}
    \mathcal{E}(\cdot) = M\,\cdot \, M^\dagger.
\end{equation}
We suppose that $\mathcal{E}$ is implemented by
block encoding $M$. That is, we have an $(n+m)$-qubit unitary $U_M$ such that $M=\alpha(\bra{1}^{\otimes m} \otimes I^{\otimes n}) U_M (\ket{1}^{\otimes m} \otimes I^{\otimes n})$ for $\alpha\geq\opnorm{M}$ and $n$, $m$ the number of system qubits and ancillas, respectively. Projecting on the $\ket{1}^{\otimes m}$ subspace by postselecting on the
all-ones outcome of the ancillas implements the map $\mathcal{E}$. Block-encoding constructions are thoroughly discussed in Refs.~\cite{Low_2019, gilyenQuantumSingularValue2019} and references therein. To keep notation light, in Eq.~\eqref{eq:cptp_map_1} we assume postselection on a single ancilla. This is without loss of generality, since we can always add an extra qubit to the block encoding and set it to $\ket{1}$ if the state of the other ancillas is $\ket{1}^{\otimes m}$.

\begin{figure}[t]
\begin{algorithm}[H]
\caption{The probabilistic quantum algorithm. Executed with \mbox{\textsc{SampleState}$(\rho_0, 0$)}. Returns the quantum state given in Lemma~\ref{lemma:expected_state} with stopping time in Lemma~\ref{lemma:stopping_time}.}
\label{alg:cap}
\begin{algorithmic}
\Require State-preparation circuit for $\rho_0$, trace nonincreasing CP map $\mathcal{E}$, parameters $T$, $\Set{r_i}_{i=0}^{T-1}$, and $r_T=1$
\Function{SampleState}{$\rho, i$}
    \State sample Bernoulli random variable $b \sim \mathcal{B}(r_i)$   
    \If{$b = 1$} \Return $\rho$ \EndIf
    \State apply channel $\Phi(\rho) = \dyad{0} \otimes \mathcal{F}(\rho) + \dyad{1} \otimes \mathcal{E} (\rho).$
    \State measure ancilla $a \sim \mathcal{B}(\tr[\mathcal{E}(\rho)])$ 
    \If{$a = 1$}
        \State $\rho \gets \frac{\mathcal{E}(\rho)}{\tr[\mathcal{E}(\rho)]}$
        \State $i \gets i+1$
    \Else
        \State $\rho \gets \rho_0$
        \State $i \gets 0$
    \EndIf
    \State \Return \Call{SampleState}{$\rho, i$}    
\EndFunction
\end{algorithmic}
\end{algorithm}
\end{figure}

\section{Results}

The proposed algorithm is summarized as pseudocode in Algorithm~\ref{alg:cap}
and the resulting stochastic process is illustrated in
Fig.~\ref{fig:algorithm_diagram}. Here we focus on three applications: solving
the discrete- and continuous-time Lyapunov equations and matrix inversion.

We introduce a deterministic stopping rule where the algorithm returns if it has
successfully completed $T$ iterations. This can be implemented simply by
setting $r_T = 1$. The
expected output state of the stopped process is captured in the following
lemma, which is shown in Appendix~\ref{app:exp_state}.
\begin{lemma}[Expected output state]\label{lemma:expected_state}
    Let $\Set{\mathcal{E},
    \mathcal{F}}$ be a quantum instrument with $\mathcal{E},
    \mathcal{F}$ CP maps between states on Hilbert space $\mathcal{H}$.
    Let $0 \leq r_k < 1$ for $k \in [T]$, $r_T=1$,
    and $\rho_0$ an arbitrary initial state on $\mathcal{H}$. Then the expected
    state at stopping time $\tau$ output by Algorithm~\ref{alg:cap} is given by
    \begin{equation}\label{eq:expected_rho}
        \mathbb{E}[\rho_{}] = \frac{\sum_{k=0}^T r_k R_k \mathcal{E}^k(\rho_0)}{\sum_{l=0}^T r_l R_l \tr[\mathcal{E}^l(\rho_0)]}
    \end{equation}
    with $R_k \coloneq \Pi_{j=0}^{k-1} (1 - r_j)$ for $k>0$ and $R_0=1$.
\end{lemma}

\begin{figure}[t]
    \centering
    \includegraphics[width=\linewidth]{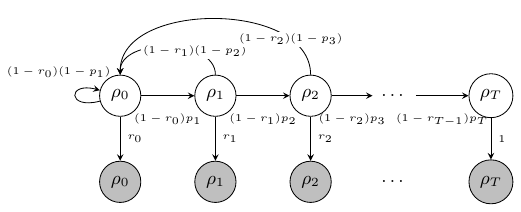}
    \caption{\textbf{Stochastic process of Algorithm~\ref{alg:cap}}. At each iteration $i$ we have state $\rho_i=\frac{\mathcal{E}^i(\rho_0)}{\tr[\mathcal{E}^i(\rho_0)]}$ (white circle), probability $r_i$ of returning the current state (grey) and stopping or probability $1-r_i$ of continuing. Conditioned on continuing we have probability $p_{i+1} = \tr[\mathcal{E}(\rho_i)]$ of applying the CP map $\mathcal{E}$ and probability $1-p_{i+1}$ of restarting with $\rho_0$ at $i=0$.}\label{fig:algorithm_diagram}
\end{figure}

We can generate different expected output states by tuning the probabilities $r_k$ and
changing $T$. Suppose we want to produce expected states with positive coefficients 
$r_k R_k = c_k > 0$ for all $k$. This can be achieved by setting the
Bernoulli probabilities to
\begin{align}
    r_k = \frac{c_k}{1 - \sum_{j=0}^{k-1} c_j} \text{ for } k \in [T].\label{eq:general_coefficients}
\end{align}
Since $r_T = 1$, we have the condition that $\sum_{k=0}^T c_k = 1$. Note that this setting is valid for $T\rightarrow\infty$ if the series of $c_k$ converges to $1$. In the following applications we work with finite $T$ and set all coefficients to a constant $c_k=1/(T+1)$. Then they cancel from Eq.~\eqref{eq:expected_rho} and we have probabilities $r_k = 1/(T+1-k)$ for all $k \in [T+1]$.

In principle, the algorithm could run forever, \eg\ in the case where
all iterations lead to a restart in Fig.~\ref{fig:algorithm_diagram}. The following
lemma shows that this does not happen on average because the expected stopping
time is bounded.
\begin{lemma}[Expected stopping time]\label{lemma:stopping_time}
    The expected stopping time of Algorithm~\ref{alg:cap} is given by
    \begin{align}
        \mathbb{E}[\tau]
        &=
        \frac{ \sum_{k=0}^{T} R_k \tr[\mathcal{E}^k(\rho_0)] } {\sum_{l=0}^{T} r_l R_l \tr[\mathcal{E}^l(\rho_0)]}.
    \end{align}
\end{lemma}
\noindent As shown in Appendix~\ref{app:exp_time}, the following corollary follows immediately.
\begin{corollary}\label{cor:stopping_time}
    Let $r_kR_k=c_k>0$ as in Eq.~\eqref{eq:general_coefficients}. Then the expected stopping time is bounded by
    \begin{align}
        \mathbb{E}[\tau]
        &\leq
        \frac{1}{\min\Set{c_k}}.
    \intertext{For the case
    $r_k = \frac{1}{T+1-k}$ with $k \in [T+1]$ the expected stopping time is
    bounded by}
        \mathbb{E}[\tau]
        &\leq
        T + 1.
    \end{align}
\end{corollary}
Due to Markov's inequality, $\mathbb{P}(\tau \geq a \mathbb{E}[\tau]) \leq \frac{1}{a}$. For example, there is a 75\% chance that the algorithm stops within 4 times the expected stopping time.

When the process stops, the resulting state can be used for downstream tasks. In Appendix~\ref{app:expectation_values} we show how to estimate expectation values, state overlaps, and normalization constants.

\subsection{Discrete-time Lyapunov equation}
Consider $N\times N$ matrices $A, B$ and recall the discrete-time Lyapunov
equation~\eqref{eq:lyapunov_discrete}, $AXA^\dagger - X + B = 0$.
We assume that $\opnorm{A} < 1$, $[A, A^\dagger] = 0$, $B$ is nonzero positive semidefinite, denoted with $B\succeq 0$, and $\tr(B) = 1$. For simplicity we assume $N=2^n$
with $n$ the number of qubits of the main register representing the desired state. Since $\opnorm{A} < 1$, a unique solution to Eq.~\eqref{eq:lyapunov_discrete}
exists and can be written as~\cite{gajicLyapunovMatrixEquation2008}
\begin{equation}
    X = \sum_{k=0}^\infty A^k B A^{\dagger k}.    
\end{equation}

Our algorithm can produce a normalized approximate solution as its expected
state of the stopped process. To see this consider the CP map
$\mathcal{E}(\cdot)= A \,\cdot \, A^\dagger$, probabilities $r_k =
\frac{1}{T+1-k}$, and initial state $\rho_0 = B$. We assume that we have a state-preparation circuit for $B$. The map $\mathcal{E}$ is implemented by block encoding its Kraus operator $A$. Inserting these parameters into
Eq.~\eqref{eq:expected_rho} shows that our algorithm produces the
expected output state
\begin{equation}\label{eq:expected_rho_discrete_lyap}
    \mathbb{E}[\rho_{}]
    =
    \frac{\sum_{k=0}^{T} A^k B A^{\dagger k}}{\sum_{k=0}^{T} \tr (A^k B A^{\dagger k})}.
\end{equation}

Next we need to bound the approximation error from finite $T$. This allows us to
choose $T$ such that the error becomes arbitrarily small as summarized by the
following theorem.
\begin{theorem}[Solution of discrete-time Lyapunov equation]\label{thm:discrete_lyapunov}
    Let $A, B \in \mathbb{C}^{2^n\times
    2^n}$ with $\opnorm{A} < 1$, $[A, A^\dagger]=0$, $B \succeq 0$, $\tr(B)=1$,
    and $0 < \epsilon < 1$. If $\opnorm{A}\neq 0$, set
    \begin{equation}\label{eq:T_discrete_lyap}
        T^* \coloneq \ceil*{ \frac{ \ln (\frac{1}{\epsilon}) } {2 \ln (\frac{1}{\opnorm{A}}) } }
    \end{equation}
    and $T^*=0$ otherwise. Instantiate Algorithm~\ref{alg:cap} with $T\geq T^*$,
    probabilities $r_k=\frac{1}{T+1-k}$ for $k \in [T+1]$, the CP map
    $\mathcal{E}(\cdot) = A\,\cdot \, A^\dagger$, and initial state $\rho_0 = B$.
    Then the expected output state satisfies $\frac{1}{2}
    \tracenorm*{\mathbb{E}[\rho_{}] -\frac{X}{\tr(X)}} \leq \epsilon$, where
    $X$ is the solution of the discrete-time Lyapunov
    equation~\eqref{eq:lyapunov_discrete}.
\end{theorem}
The proof is in Appendix~\ref{app:proof_error_discrete}. In Appendix~\ref{app:optimality_of_Tstar} we prove that the choice of $T^*$ in Theorem~\ref{thm:discrete_lyapunov} is asymptotically optimal.

Choosing $T=T^*$ results in the shortest expected stopping time. Then Corollary~\ref{cor:stopping_time} implies that the algorithm produces an approximation to the normalized solution of the discrete-time Lyapunov equation with (in expectation) at most  $T^*+1$ calls to the block encoding of $A$ and the state-preparation circuit for $B$.

\subsection{Continuous-time Lyapunov equation}
Now recall the continuous-time Lyapunov equation~\eqref{eq:lyapunov_continuous}, $AX + XA^\dagger + B = 0$.
In this subsection, we assume that all eigenvalues of $A$ have strictly negative real parts. Furthermore, $[A, A^\dagger]=0$, $B \succeq 0$, and $\tr(B) = 1$ as before. Eq.~\eqref{eq:lyapunov_continuous} has a unique solution if and only if the eigenvalues of $A$ satisfy $\lambda_i + \lambda_j^* \neq 0$ for all $\lambda_i, \lambda_j \in \lambda(A)$. This is satisfied because all eigenvalues have strictly negative real parts. Furthermore, the solution is given by~\cite{gajicLyapunovMatrixEquation2008}
\begin{align}
\label{eq:sol_cont_time_lyapunov}
    X = \int_0^\infty e^{tA} B e^{tA^\dag} dt.
\end{align}

To approximate this solution with our algorithm we truncate the integral in Eq.~\eqref{eq:sol_cont_time_lyapunov} to the domain $[0, L]$ and discretize using the left Riemann sum. We choose an interval size $\Delta$ and, for integer $T$, a cutoff $L = (T+1)\Delta$. Then
\begin{align}
\label{eq:riemann_sum}
    X \approx \Delta \sum_{k=0}^{\frac{L}{\Delta} - 1} e^{k \Delta A} B e^{k \Delta A^\dagger}.
\end{align}
Let us choose as the CP map $\mathcal{E}(\cdot) = e^{\Delta A}\,\cdot \, e^{\Delta A^\dagger}$, probabilities $r_k = \frac{1}{T+1-k}$, and initial state $\rho_0=B$. For the moment we assume that $\mathcal{E}$ is implemented exactly by block encoding its Kraus operator $e^{\Delta A}$. Inserting the parameters into Eq.~\eqref{eq:expected_rho} shows that our algorithm produces the
expected output state
\begin{equation}\label{eq:expected_rho_continuous_lyap}
    \mathbb{E}[\rho_{}]
    =
    \frac{\sum_{k=0}^{T} e^{k\Delta A} B e^{k\Delta A^\dagger}}{\sum_{k=0}^{T} \tr (e^{k\Delta A} B e^{k\Delta A^\dagger})}.
\end{equation}
It remains to bound the approximation error from the discretization (denoted $\epsilon_1$) and truncation ($\epsilon_2$)
of the solution. This allows us to choose parameters $T$ and
$\Delta$ such that the error becomes arbitrarily small as summarized by the
following theorem proven in Appendix~\ref{app:proof_error_continuous}.

\begin{theorem}[Solution of continuous-time Lyapunov equation]\label{thm:continuous_lyapunov}
    Let $A, B \in \mathbb{C}^{2^n\times 2^n}$ with $\Re(\lambda_i) < 0$ for all $\lambda_i\in \lambda(A)$, $[A,
    A^\dagger]=0$, $B \succeq 0$, $\tr(B)=1$, and $0 < \epsilon_1, \epsilon_2 < 1$. Set
    \begin{align}
        T^*
        &\coloneq \ceil*{ \frac{1}{2\epsilon_1} \ln( \frac{1}{\epsilon_2}) \frac{ \opnorm{A}\abs*{\min_k \Set{ \Re(\lambda_k) }}}{ \abs*{\max_j \Set{ \Re(\lambda_j) }}^2 } } , \label{eq:T_continuous_lyap}\\
        \Delta^*
        &\coloneq \frac{\epsilon_1}{\opnorm{A}}\frac{\max_j\Set{\Re(\lambda_j)}}{\min_k\Set{\Re(\lambda_k)}}.
        \label{eq:Delta_continuous_lyap}
    \end{align}
    Instantiate Algorithm~\ref{alg:cap} with $T \geq T^*$, $0 < \Delta \leq \Delta^*$,  probabilities $r_k=\frac{1}{T+1-k}$ for $k \in [T+1]$, the CP map $\mathcal{E}(\cdot) = e^{\Delta A} \,\cdot \,
    e^{\Delta A^\dagger}$, and initial state $\rho_0 = B$. Then the expected output state satisfies $\frac{1}{2} \tracenorm*{\mathbb{E}[\rho_{}] - \frac{X}{\tr(X)}} \leq \epsilon_1 + \epsilon_2$, where $X$ is the solution of the continuous-time Lyapunov equation~\eqref{eq:lyapunov_continuous}.
\end{theorem}

Corollary~\ref{cor:stopping_time} implies that the algorithm requires in expectation at most $T^*+1$ calls to the block encoding of $e^{\Delta A}$ and the state-preparation circuit for $B$. This is multiplied by the cost of implementing the block encoding, which, of note, shall not have subnormalization factors that can bias the transition probabilities $r_k$.
In the special case of Hermitian $A$, the condition number is $\kappa=\frac{\min_k \Set{\lambda_k}}{\max_j\Set{\lambda_j}}$, and thus we find $T^* = \ceil*{ \frac{\kappa^2}{2\epsilon_1}\ln(\frac{1} {\epsilon_2}) }$ and $\Delta^* = \frac{\epsilon_1}{\kappa\opnorm{A}}$. 

\subsection{Matrix inversion}

Consider a matrix $A \succ 0$ and the task of finding its normalized inverse. We use $\kappa$ to indicate its condition number. We can reduce the problem to the discrete-time Lyapunov equation $A' X (A')^\dag - X + B = 0$ with $A' = \sqrt{I-A/\opnorm{A}}$ and $B = I/N$. To see this, note that $X = A^{-1} \opnorm{A}/N$ commutes with $A'$ and solves the equation. Since $\opnorm{A'} = \sqrt{1 - 1/\kappa} < 1$, the solution is unique. Finding the solution with our algorithm requires the CP map $\mathcal{E}(\cdot) = \sqrt{I - A/\opnorm{A}} \,\cdot \, \sqrt{I - A/\opnorm{A}}$. Plugging this map in Theorem~\ref{thm:discrete_lyapunov} we obtain a state that satisfies $\frac{1}{2}\tracenorm*{\mathbb{E}[\rho_{}] - \frac{A^{-1}}{\tr(A^{-1})}}\leq \epsilon$ and expected stopping time bounded by $\ceil*{ \kappa \ln(1/\epsilon)}+1$ yielding query complexity $\mathcal{O}[\kappa \ln(1/\epsilon)]$ in expectation~\footnote{Here we use that $1/\ln(1/\opnorm{A'}) = 1/\ln(1/\sqrt{1 - 1/\kappa}) < 2 \kappa$ for $\kappa >1$.}. 

An alternative path to matrix inversion is to reduce it to the continuous-time Lyapunov equation $A'X + X(A')^\dag +B =0$ with $A' = -A/2$ and $B=I/N$. Note that $X = A^{-1}/N$ solves the equation. Since $\Re(\lambda'_i) = - \Re(\lambda_i)/2<0$ for any pair of eigenvalues $\lambda'_i$ and $\lambda_i$ corresponding to the same eigenvector of $A'$ and $A$, respectively, the solution is unique. Using Theorem~\ref{thm:continuous_lyapunov} for this case with the choice $\epsilon_1=\epsilon_2=\epsilon/2$, we obtain an expected output state that satisfies $\frac{1}{2}\tracenorm*{\mathbb{E}[\rho_{}] - \frac{A^{-1}}{\tr(A^{-1})}}\leq \epsilon$ using the CP map $\mathcal{E}(\cdot) = e^{-\frac{\epsilon A}{2\kappa \opnorm{A}}} \, \cdot \, e^{-\frac{\epsilon A}{2\kappa \opnorm{A}}}$. The expected stopping time is bounded by $\ceil*{ \frac{\kappa^2}{\epsilon} \ln(\frac{2}{\epsilon}) }+1$.

In Table~\ref{tab:summary} we denote the discrete-time approach as matrix Inversion I, and the continuous-time approach as matrix Inversion II. Which of the two algorithms is preferable? The discrete-time approach has a favorable scaling in both $\epsilon$ and $\kappa$, but the answer really depends on the cost of the block encoding. If we can block encode $\sqrt{A/\opnorm{A}}$, then the qubitization trick~\cite[Lemmas 8 and 10]{Low_2019} yields the block encoding of $\sqrt{I -A/\opnorm{A}}$ in an off-diagonal block, which could be used in our discrete-time algorithm. Alternatively, one can use the quantum singular value transformation~\cite{gilyenQuantumSingularValue2019} to block-encode approximations to functions of $A$, including $\sqrt{I -A/\opnorm{A}}$ and $e^{-\frac{\epsilon A}{2\kappa \opnorm{A}}}$ (see Appendix~\ref{app:block_encoding_exp} for an explicit construction of the latter).

\subsection{Approximate implementation}
So far we have assumed an exact implementation of the CP map $\mathcal{E}$ via a block encoding of $M$. Now suppose that we can only implement an approximate block encoding $\tilde M$ yielding the trace nonincreasing CP map $\tilde{\mathcal{E}}(\cdot) = \tilde M\,\cdot \, \tilde M^\dagger$. This leads to a further additive error $\epsilon_{\sim}$ in the trace distance $\frac{1}{2} \tracenorm*{\mathbb{E}[\rho_{}] - \frac{X}{\tr(X)}}$. As shown in Appendix~\ref{app:robustness} this error can be achieved in practice by choosing a block encoding such that $\opnorm{M - \tilde{M}} \leq \frac{\epsilon_\sim}{2T(T+1)}$ with $T$ from Theorem~\ref{thm:discrete_lyapunov} or Theorem~\ref{thm:continuous_lyapunov}, respectively. For example, the error of the continuous-time algorithm with this approximate implementation is $\frac{1}{2} \tracenorm*{\mathbb{E}[\rho_{}] - \frac{X}{\tr(X)}} \leq \epsilon_1 + \epsilon_2 + \epsilon_{\sim}$. In Appendix~\ref{app:block_encoding_exp} we provide an explicit construction of the approximate block encoding of $M=e^{\Delta A}$, using $\mathcal{O}\left[\kappa \log (\frac{T}{\epsilon_\sim})\right]$ calls to the block encoding of $A$.

\section{Conclusion}

We have presented a probabilistic quantum algorithm for approximating solutions to Lyapunov equations and matrix inversion under certain conditions. Given block-encoding access to a function of the normal input matrix $A$ and a state-preparation circuit for $\rho_0$, in expectation the algorithm outputs mixed states $\propto \sum_{k=0}^T \mathcal{E}^k(\rho_0)$. With the choice of $T$, $\mathcal{E}$, and $\rho_0$ in Table~\ref{tab:summary} the expected output state is proportional to the respective solution up to the given error in trace distance. The expected stopping time of the algorithm is at most $T+1$ yielding query complexity $\mathcal{O}(T)$. For positive definite $A$ and access to a block encoding of $\sqrt{A/\opnorm{A}}$ the query complexity of matrix inversion matches the lower bound $\Omega[\kappa\ln(1/\epsilon)]$ of quantum linear system solvers~\cite{costaDiscreteAdiabaticQuantum2025,orsucciSolvingClassesPositivedefinite2021,mori2026sparsity}. In contrast, however, we only get sample access to the expected output state representing the solution. Nevertheless, our algorithm can be used in applications that require estimating expectations, overlaps, or traces of the normalized solution.

The presented algorithm could be extended to higher-precision quadrature formulas for Eq.~\eqref{eq:sol_cont_time_lyapunov} and to approximating more general integrals than Eq.~\eqref{eq:sol_cont_time_lyapunov}. First, admitting a different CP map at each iteration yields an expected output state $\propto \sum_{k=0}^T c_k (\mathcal{E}_k\circ\cdots\circ \mathcal{E}_0)(\rho_0)$. For example, $\mathcal{E}_k(\cdot) = e^{(t_k-t_{k-1})A} \, \cdot \, e^{(t_k-t_{k-1})A^\dagger}$ with a suitable choice of $t_k$, $c_k$ leads to a matrix equivalent of Gaussian quadrature~\cite{sinapPolynomialInterpolationGaussian1994}. For smooth scalar functions Gaussian quadrature rules typically converge much faster than the Riemann sum we use. It would be interesting to analyze whether this can improve the scaling of parameter $T$ with $\epsilon_1$ in the continuous-time Lyapunov case. Second, in Eq.~\eqref{eq:sol_cont_time_lyapunov} we could introduce a weighting function $f(t)$ or admit an anti-Hermitian $A$. The former leads to a generalization of the Laplace transform to matrices, the latter to the (one-sided) operator Fourier transform. It would be fruitful to analyze approximation errors and applications of those generalizations. 

\bigskip
\paragraph*{Acknowledgments.}
We thank Hitomi Mori, Maria Tudorovskaya, and Alexandre Krajenbrink for helpful feedback, and Frederic Sauvage for useful discussions.

\onecolumngrid
\appendix

\section{Incomparability of block encodings and states as resources}
\label{app:incompatibility}

Consider a positive semidefinite matrix $C\in\mathbb{C}^{N\times N}$ of spectral norm $\opnorm{C}=1$. This data can be encoded as a state $\rho_C=C/\tr[C]$, or as a block encoding $U_C$ such that $\left(I\otimes\bra{G}\right)U_C\left(I\otimes\ket{G}\right)=C$ for some known state $\ket{G}$. The following two theorems, which are adapted from Ref.~\cite{Somma_2025}, show that the two encodings are incomparable computational resources.

\begin{theorem}
For a permutation matrix $P=\sum_{j=1}^N\ket{P(j)}\!\bra{j}\in\mathbb{C}^{N\times N}$, let $C=(2I+P+P^\dag)/4$. Consider the decision problem of determining whether $\bra{1}P\ket{1}$ is $1$ or $0$, that is, respectively, whether $\bra{1}C\ket{1}$ is $1$ or $1/2$. This problem can be solved with high probability using a constant number of queries to the block encoding $U_C$, but it requires $\Omega(\sqrt{N})$ copies of the mixed state $\rho_C$.
\end{theorem}

Note that both $P$ and its inverse $P^\dag$ yield the same $C$ and, in the reverse direction, $C$ uniquely determines the unordered pair $\Set{P,P^\dag}$ (which is a singleton if $P$ is its own inverse). Hence the problem is well defined.

\begin{proof}
Consider an algorithm that $\ell$ times repeats the following: it applies $U_C$ to $\ket{1}\ket{G}$ and then performs the two-outcome projective measurement corresponding to $\ket{1}\ket{G}$ and its orthogonal complement. If all $\ell$ measurements yield $\ket{1}\ket{G}$ the algorithm returns $\innerp{1}{P}{1}=1$, otherwise it returns $\innerp{1}{P}{1}=0$. Since the algorithm essentially samples from the probability distribution $\left(\innerp{1}{C}{1}^2,\,1-\innerp{1}{C}{1}^2\right)^\ell$, it fails with probability at most $1/4^\ell$. Hence, for any constant error $\epsilon>0$ a constant number of queries to $U_C$ suffices. 

For the lower bound on the copies of $\rho_C$, we will employ Ref.~\cite[Theorem 2]{Somma_2025}. This tells us that, if we are given asymptotically fewer than $\sqrt{N}$ copies of the state $\ket{\phi_P} \coloneq \frac{1}{\sqrt{N}} \sum_{j=1}^N\ket{j,P(j)}$, then we cannot with high probability decide if $P(1)=1$. Thus, it suffices to show that there is a random procedure, which, in expectation, uses a constant number of copies of $\ket{\phi_P}$ and prepares the state $\rho_C$. Let $F_P \coloneq \Set{j\given P(j)=j}$ be the set of fixed points of $P$ and let $M_P \coloneq \abs{F_P}=\tr[P]=\tr[P^\dag]$. Note that
\begin{align}
\rho_C
& = \frac{2I+P+P^\dag}{2N+2M_P}  \\
& =
\frac1{2N+2M_P}
\left(
\sum_{j\in F_P}4\dyad{j}
+\sum_{j\notin F_P}\Big(2\dyad{j}+\ketbra{P(j)}{j}+\ketbra{P^\dag(j)}{j}\Big)
\right)
 \\& =
\frac1{2N+2M_P}
\left(
\sum_{j\in F_P}4\dyad{j}
+\sum_{j\notin F_P}\Big(
\dyad{j}+\dyad{P(j)}+\ketbra{P(j)}{j}+\ketbra{j}{P(j)}
\Big)
\right)
 \\
\label{eq:rhoC4}
& =
\frac{2N}{N+M_P}
\left(
\frac{1}{N}
\sum_{j\in F_P}\dyad{j}
+\frac{1}{2N}\sum_{j\notin F_P}\frac{\ket{j}+\ket{P(j)}}{\sqrt{2}}\frac{\bra{j}+\bra{P(j)}}{\sqrt{2}}
\right).
\end{align}
Now we show that, given the expected number of $2/(1+M_P/N)\in[1,2]$ copies of the state $\ket{\phi_P}$, we can prepare one copy of $\rho_C$. First, we measure $\ket{\phi_P}$ in the computational basis, obtaining $\ket{j,P(j)}$ for a uniformly random $j$. If $P(j)=j$, we prepare and return the state $\ket{j}$. If $P(j)\ne j$, with probability $1/2$ we prepare and return the state $(\ket{j}+\ket{P(j)})/\sqrt{2}$ and with probability $1/2$ we restart the procedure with a fresh copy of $\ket{\phi_P}$.

To see that this procedure indeed gives the state $\rho_C$ in expectation, consider Eq.~\eqref{eq:rhoC4}. The term $\frac{1}{N}
\sum_{j\in F_P}\dyad{j}$ corresponds to the event when the measurement of $\ket{\phi_P}$ yields $\ket{j,j}$. The trace of this former term is $\frac{M_P}{N}$, which is the probability of this event. The other term, namely, 
\begin{equation}
\frac12\cdot\frac{1}{N}\sum_{j\notin F_P}\frac{\ket{j}+\ket{P(j)}}{\sqrt{2}}\frac{\bra{j}+\bra{P(j)}}{\sqrt{2}},
\end{equation}
corresponds to the event when \emph{both} the measurement of $\ket{\phi_P}$ yields $\ket{j,P(j)}$ with $P(j)\ne j$ \emph{and}, subsequently, we decide with probability $1/2$ not to restart the procedure. The trace of this latter term is $\frac{N-M_P}{2N}$, which is the probability of this event.
Hence, the unnormalized mixed state in the parenthesis of Eq.~\eqref{eq:rhoC4} is the output of the procedure conditioned on having concluded without a restart. The trace of this state is $\frac{N+M_P}{2N}$, which is the probability of the procedure concluding without a restart, therefore we need to multiply by its reciprocal to normalize the state, thus obtaining $\rho_C$.
Note that $\frac{N+M_P}{2N}=1-\frac{N-M_P}{2N}$, where $\frac{N-M_P}{2N}$ is the probability that we restart the procedure.
\end{proof}

\begin{theorem}
Let $C\in\mathbb{C}^{N\times N}$ be a binary diagonal matrix with exactly one nonzero entry, that is, $C=\dyad{j}$ for some unknown $j\in\Set{1,\ldots,N}$. Consider the search problem of finding the unique $j$ such that $\innerp{j}{C}{j}=1$. This problem requires $\Omega(\sqrt{N})$ queries to the block encoding $U_C$, but can be solved while a single copy of the state $\rho_C$.
\end{theorem}

\begin{proof}
The lower bound $\Omega(\sqrt{N})$ on the number of queries to the block encoding is given by Ref.~\cite[Theorem 4]{Somma_2025}. On the other hand, because $\rho_C=C=\dyad{j}$, measuring the unique copy of $\rho_C$ in the computational basis yields $j$ with certainty.
\end{proof}

Arguably, an even stronger evidence of incomparability of $U_C$ and $\rho_C$ as computational resources would be to give two computational problems for the same $C$, such that for one problem $U_C$ is a more useful resource while $\rho_C$ is more useful for the other. Whether such $C$ exists is left as an open problem.

\section{Analysis of the probabilistic algorithm}

The algorithm is described in the main text and presented as pseudocode in Algorithm~\ref{alg:cap}. Here we show one iteration of the algorithm as a transformation, while hiding all the intermediary steps such as coin flip, block encoding, and ancilla measurement:
\begin{align}
\label{eq:intuition}
    \rho_i \; \longrightarrow \; r_i \underbrace{\rho_i}_\text{return} + \; (1 - r_i) \tr[\mathcal{E}(\rho_i)] \;  \underbrace{\frac{\mathcal{E}(\rho_i)}{\tr [\mathcal{E}(\rho_i)]}}_{\substack{ =: \; \rho_{i+1} \\ \text{continue}}} + \; (1-r_i) (1 - \tr [\mathcal{E}(\rho_i)]) \underbrace{\rho_0}_\text{restart}.
\end{align}

\subsection{Expected output state (Lemma~\ref{lemma:expected_state})}
\label{app:exp_state}

We calculate the expected output state $\mathbb{E}[\rho_{}]$ of the algorithm following the derivation in Zhang et al.~\cite{zhang2023dissipative}. Here $\rho_{}$ denotes the quantum state of the stopped process.
We begin by noting that whenever the algorithm restarts, it does so from the initial state $\rho_0$ and the iteration counter is reset to $i=0$. In this case the whole probability subtree can be replaced by $\mathbb{E}[\rho_{}]$. 
We have
\begin{align}
    \mathbb{E}[\rho_{}] &= r_0 \rho_0 + (1-r_0) \left(1 - \tr [\mathcal{E}(\rho_0)] \right) \mathbb{E}[\rho_{}]
    \\ 
    &
    + r_1 (1- r_0) \mathcal{E}(\rho_0) + (1 - r_1) (1 - r_0) \left(\tr[\mathcal{E}(\rho_0)] - \tr[\mathcal{E}^2(\rho_0)] \right) \mathbb{E}[\rho_{}]
    \\
    &
    + \dots
    \\
    &
    + \left[r_{k} \prod_{j=0}^{k-1} (1 - r_j) \right] \mathcal{E}^{k}(\rho_0) + \left[ \prod_{i=0}^{k} (1 - r_i) \right] \left(\tr[\mathcal{E}^{k}(\rho_0)] - \tr[\mathcal{E}^{k+1}(\rho_0)] \right) \mathbb{E}[\rho_{}]
    \\
    &
    + \dots
    \\
    &
    + \left[ \prod_{l=0}^{T-1} (1 - r_l) \right] \mathcal{E}^{T}(\rho_0) ,
\end{align}
where the last line is given by the deterministic stopping rule, $r_T = 1$. Grouping terms and introducing the notation $R_k = \prod_{l=0}^{k-1} (1 - r_l)$ and $R_0 = 1$ we have 
\begin{align}
    \mathbb{E}[\rho_{}]
    &=
    \sum_{k=0}^{T} r_k R_k \mathcal{E}^k(\rho_0) + \mathbb{E}[\rho_{}] \sum_{l=0}^T R_{l+1} \left\{\tr[\mathcal{E}^l(\rho_0)] - \tr[\mathcal{E}^{l+1}(\rho_0)] \right\}.
\end{align}
With some work one can show that the second summation is equivalent to $1 - \sum_{l=0}^{T} r_l R_l \tr[\mathcal{E}^{l}(\rho_0)]$. Solving for the expected state we arrive at
\begin{align}
\label{eq:expected_rho_app}
    \mathbb{E}[\rho_{}]
    =
    \frac{ \sum_{k=0}^{T} r_k R_k \mathcal{E}^k(\rho_0) } {\sum_{l=0}^{T} r_l R_l \tr[\mathcal{E}^l(\rho_0)]} .
\end{align}
By tuning the probabilities $r_k$ and the stopping parameter $T$ it is now possible to implement different expected states. As a warm-up example, we consider the case where $r_k R_k = c$ for some constant $c > 0$ and for all $k$, thus canceling out from Eq.~\eqref{eq:expected_rho_app}. This can be achieved by setting the Bernoulli probabilities to $r_k = \frac{c}{1 - k c}$. Since $r_T = 1$ we find that $c = \frac{1}{T+1}$. Thus, we have
\begin{align}
\label{eq:constant_coefficients_app}
    r_k
    =
    \frac{1}{1 + T  - k} \text{ for } k \in [T+1]
    \quad \implies \quad
    \mathbb{E}[\rho_{}]
    =
    \frac{\sum_{k=0}^T \mathcal{E}^k (\rho_0)}{\sum_{l=0}^T \tr [\mathcal{E}^{l} (\rho_0)]} .
\end{align}
We emphasize that this setting is valid only for finite $T$ since, otherwise, all $r_k = 0$, \ie\ the algorithm never enters the `return' branch in Eq.~\eqref{eq:intuition}. As a second example, we consider the more general case where $r_k R_k = c_k$ for coefficients $0< c_k < 1$. This can be achieved by the setting
\begin{align}
\label{eq:series_coefficients_app}
    r_k
    =
    \frac{c_k}{1 - \sum_{j=0}^{k-1} c_j} \text{ for } k \in [T+1]
    \quad \implies \quad
    \mathbb{E}[\rho_{}]
    =
    \frac{\sum_{k=0}^T c_k \mathcal{E}^k (\rho_0)}{\sum_{l=0}^T c_l \tr [\mathcal{E}^{l} (\rho_0)]} .
\end{align}
Since $r_T = 1$, we have the condition that $\sum_{k=0}^T c_k = 1$, \ie\ the coefficients must form a probability vector. If $\Set{c_k}$ also forms a convergent series then, unlike the previous example, $T$ can be taken to infinity.

\subsection{Expected stopping time (Lemma~\ref{lemma:stopping_time} and Corollary~\ref{cor:stopping_time})}
\label{app:exp_time}

For the expected stopping time $\mathbb{E}[\tau]$ of the algorithm one simply follows the derivation in~\cite[Lemma 7]{zhang2023dissipative} to get
\begin{align}
    \mathbb{E}[\tau]
    = \frac{ \sum_{k=0}^{T} R_k \tr[\mathcal{E}^k(\rho_0)] } {\sum_{l=0}^{T} r_l R_l \tr[\mathcal{E}^l(\rho_0)]}.
\end{align}
When the Bernoulli probabilities $r_k$ are given by Eq.~\eqref{eq:series_coefficients_app}, we have $R_k = 1 - \sum_{j=0}^{k-1} c_j$. Then
\begin{align}
    \mathbb{E}[\tau]
    =
    \frac{ \sum_{k=0}^{T} \left(1 - \sum_{j=0}^{k-1} c_j\right) \tr[\mathcal{E}^k(\rho_0)] } {\sum_{l=0}^{T} c_l \tr[\mathcal{E}^l(\rho_0)]}
    \leq
    \frac{ \sum_{k=0}^{T} 
    \tr[\mathcal{E}^k(\rho_0)] } {\sum_{l=0}^{T} c_l \tr[\mathcal{E}^l(\rho_0)]}
    \leq \frac{1}{\min \Set{c_l}},
\end{align}
where we use that $0 < c_k < 1$ and $0 < \tr[\mathcal{E}^k(\rho_0)] \leq 1$ for all $k$ by assumption, so we can ignore the term with the minus sign. It follows that for constant coefficients $c_k=\frac1{T+1}$ as in Eq.~\eqref{eq:constant_coefficients_app}, we have $\mathbb{E}[\tau] \leq T+1$. Intuitively, this is because the Bernoulli probabilities of entering the `return' branch in Eq.~\eqref{eq:intuition} are $r_k\ge \frac1{T+1}$.

\section{Approximation error for Lyapunov equations}

We start by the following claim, which we employ below in Appendices~\ref{app:proof_error_discrete}, \ref{app:proof_error_continuous}, and \ref{app:robustness}.
Let $\succeq$ denote the Loewner order.

\begin{claim}
\label{clm:normalized_semidef_dist}
For nonzero positive semidefinite $Y,Z$ of the same dimensions, we have
\begin{align}
    \frac{1}{2}\tracenorm*{\frac{Y}{\tr(Y)}-\frac{Z}{\tr(Z)}} \le \frac{\tracenorm*{Y-Z}}{\max\Set{\tr(Y), \tr(Z)}}.
\end{align}
\end{claim}

\begin{proof}
Assume that $\tr(Y)\ge\tr(Z)$ (if not one can simply swap the roles of $Y$ and $Z$ in the following derivation).
Because $Y,Z\succeq0$, we have $\tracenorm*{Y}=\tr(Y)$ and $\tracenorm*{Z}=\tr(Z)$. By applying the triangle inequality to $\tracenorm{(Y-Z)+Z}$,
we get
\begin{align}
\label{eq:dist_between_semidef}
    \tr(Y)-\tr(Z) \le \tracenorm*{Y-Z},
\end{align}
which we will use in a moment. Applying the triangle inequality to the distance that we want to bound, we get
\begin{align}
    \frac{1}{2}\tracenorm*{\frac{Y}{\tr(Y)}-\frac{Z}{\tr(Z)}}
    & \le
    \frac{1}{2}\tracenorm*{\frac{Y}{\tr(Y)}-\frac{Z}{\tr(Y)}}
    + \frac{1}{2}\tracenorm*{\frac{Z}{\tr(Y)}-\frac{Z}{\tr(Z)}}
    \\ & \le
    \frac{\tracenorm*{Y-Z}}{2\tr(Y)}
    + \frac{\tracenorm*{Z}}{2}\abs*{\frac{1}{\tr(Y)}-\frac{1}{\tr(Z)}}
    \\ & \le
    \frac{\tracenorm*{Y-Z}+[\tr(Y)-\tr(Z)]}{2\tr(Y)}
\end{align}
where we use $\tr(Y)\ge\tr(Z)$ and, once again, $\tracenorm*{Z}=\tr(Z)$. The proof is concluded by applying Eq.~\eqref{eq:dist_between_semidef}.
\end{proof}

\subsection{Proof of Theorem~\ref{thm:discrete_lyapunov} (solution of the discrete-time Lyapunov equation)}
\label{app:proof_error_discrete}

Let us recap the setting of the theorem. We want to solve the discrete-time Lyapunov equation
\begin{align}
\label{eq:disc_time_lyapunov}
    A X A^\dag - X + B = 0 ,
\end{align}
where $B$ is positive semidefinite, $\tr(B)=1$, $[A, A^\dag] = 0$, and all eigenvalues of $A$ have magnitude strictly less than $1$ (the following are equivalent: (i) it is Schur stable, (ii) its spectral radius is less than $1$). The spectral radius of a normal matrix equals the operator norm.
Let us introduce the following notation for this section:
\begin{align}
\label{eq:sol_disc_time_lyapunov}
    \mathcal{S}(T) \coloneq \sum_{k=0}^{T} A^{k} B A^{k \dag}.
\end{align}
Note that $B = C C^\dag$ for some $C$, since $B$ is positive semidefinite. Then the summand in $\mathcal{S}(T)$ can be written in the form $(A^{k} C)(A^{k}C)^\dag$ and so it is also positive semidefinite.
It follows that $\mathcal{S}(T)$ itself is positive semidefinite and thus
\begin{align}
    & \tracenorm*{\mathcal{S}(T)} = \tr[\mathcal{S}(T)]  \\
    & \mathcal{S}(N) \succeq  \mathcal{S}(T) \text{ for } N \geq T  \\
    \label{eq:SInf_ST_norm}
    & \tracenorm*{ \mathcal{S}(\infty) - \mathcal{S}(T) } = \tr[\mathcal{S}(\infty)] - \tr[\mathcal{S}(T)].
\end{align}
Equation~\eqref{eq:disc_time_lyapunov} has a unique solution $X$ if and only if the
eigenvalues of $A$ satisfy $\lambda_i \lambda_j \neq 1$ for all $i$, $j$. This
is satisfied because we assume $\opnorm{A} < 1$. Furthermore, for $\opnorm{A} < 1$ the solution to Eq.~\eqref{eq:disc_time_lyapunov} is $X =
\mathcal{S}(\infty)$.
We can produce a normalized version of Eq.~\eqref{eq:disc_time_lyapunov} by using the probabilistic quantum algorithm with the following setup: the block encoding of $A$ is used to implement the map $\mathcal{E}(\cdot) = A \,\cdot \, A^\dag$, and the Bernoulli probabilities given in Eq.~\eqref{eq:constant_coefficients_app} are used to implement constant coefficients $c_k$. Then the expected output state is 
\begin{align}
\label{eq:lyapunov_algo_output}
     \mathbb{E}[\rho] = \frac{\sum_{k=0}^{T} A^k B A^{\dagger k}}{\tr (\sum_{l=0}^{T} A^l B A^{\dagger l})} =  \frac{\mathcal{S}(T)}{\tr[\mathcal{S}(T)]}.
\end{align}
To prove the theorem we find the value of the parameter $T$ such that the expected output state is $\epsilon$-close to the normalized solution in trace distance.

\begin{proof}[Proof of Theorem~\ref{thm:discrete_lyapunov}]
Using Claim~\ref{clm:normalized_semidef_dist} and, then, Eq.~\eqref{eq:SInf_ST_norm}, we get
\begin{align}
    \frac{1}{2}\tracenorm*{ \frac{X}{\tr[X]} - \mathbb{E}[\rho] } &= \frac{1}{2}\tracenorm*{ \frac{ \mathcal{S}(\infty) }{\tr[\mathcal{S}(\infty)]} - \frac{\mathcal{S}(T)}{\tr[\mathcal{S}(T)]} }
    \leq 
    \frac{ \tracenorm*{\mathcal{S}(\infty) - \mathcal{S}(T)}}{\tr[\mathcal{S}(\infty)]}
    = \frac{ \tr[\mathcal{S}(\infty)] - \tr[\mathcal{S}(T)]}{\tr[\mathcal{S}(\infty)]} .
\label{eq:derivation_2} 
\end{align}
For the numerator we find
\begin{align}
    \tr[\mathcal{S}(\infty)] - \tr[\mathcal{S}(T)]
    &= \sum_{k=T+1}^\infty \tr(A^k B A^{\dag k}) = \sum_{j=0}^\infty \tr(A^{j+T+1} B A^{\dag j+T+1}) \\
    &= \tr[(A^\dag A)^{T+1} \sum_{j=0}^\infty (A^\dag A)^j B] = \tr[(A^\dag A)^{T+1} (I - A^\dag A)^{-1} B] ,
\label{eq:derivation_3}
\end{align}
where we use the closed form of the Neumann series. Since $A^\dag A$ is positive semidefinite, we can write its spectral decomposition as $A^\dag A = \sum_i \alpha_i \dyad{i}$ where $0\le\alpha_i<1$ and $\Set{\ket{i}}_i$ is an eigenbasis of $A^\dag A$. Since $B$ is positive semidefinite, we have that $\beta_i=\innerp{i}{B}{i} \geq 0$. Thus we can continue as
\begin{align}
\label{eq:split_trace}
    \tr[\mathcal{S}(\infty)] - \tr[\mathcal{S}(T)]
    &=
    \sum_i \alpha_i^{T+1} \frac{\beta_i}{1-\alpha_i} \leq \max_j \Set*{\alpha_j^{T+1}} \sum_i \frac{\beta_i}{1-\alpha_i} \\
    &=
    \max\Set{\abs{\lambda_i}\given \lambda_i\in\lambda(A)}^{2(T+1)} \tr[(I - A^\dag A)^{-1} B],
\end{align}
with equality when $T=-1$, for which $\mathcal{S}(-1) = 0$. Recall that $A$ is normal so its spectral radius equals the operator norm. Then Eq.~\eqref{eq:derivation_2} is bounded as  
\begin{align}
\label{eq:upperbound_discrete}
    \frac{1}{2}\tracenorm*{ \frac{X}{\tr[X]} - \mathbb{E}[\rho] }
    \leq
    \frac{\tr[\mathcal{S}(\infty)] - \tr[\mathcal{S}(T)]}{\tr[\mathcal{S}(\infty)] - \tr[\mathcal{S}(-1)]}  
    \leq
    \opnorm{A}^{2(T+1)}
    \leq
    \epsilon,
\end{align}
where we introduced a parameter $0 < \epsilon < 1$. Assuming $\opnorm{A}\neq 0$ and solving for $T$ we conclude that
\begin{align}
\label{eq:lyapunov_time_and_error}
    T
    \geq
    \ceil*{ \frac{ \ln (\frac{1}{\epsilon}) } {2 \ln (\frac{1}{\opnorm{A}})} } \quad \implies \quad \frac{1}{2}\tracenorm*{ \frac{X}{\tr[X]} - \mathbb{E}[\rho] }
    \leq
    \epsilon.
\end{align}
The trivial case $\opnorm{A}=0$ has solution $X=B$, which is achieved by setting $T\geq 0$.
\end{proof}

\subsection{Proof of Theorem~\ref{thm:continuous_lyapunov} (solution of the continuous-time Lyapunov equation)}
\label{app:proof_error_continuous}

We consider the continuous-time Lyapunov equation
\begin{align}
\label{eq:cont_time_lyapunov_app}
    A X + X A^\dag + B = 0,
\end{align}
where $B$ is positive semidefinite, $\tr(B)=1$, $[A, A^\dag] =0$, and all eigenvalues of $A$ satisfy $\Re[\lambda(A)] < 0$ (Hurwitz stable).
Equation~\eqref{eq:cont_time_lyapunov_app} has a unique solution if and only if the eigenvalues of $A$ satisfy $\lambda_i + \lambda_j^* \neq 0$ for all $\lambda_i, \lambda_j\in\lambda(A)$. 
This is satisfied because we assume $\Re [\lambda(A)] < 0$ so the sum of the real parts of $\lambda_i + \lambda_j^*$ cannot be $0$. 
Furthermore, the unique solution is given by 
\begin{align}
\label{eq:target}
    X = \int_0^\infty e^{tA} B e^{tA^\dag}dt ,
\end{align}
which can be verified by plugging it back in Eq.~\eqref{eq:cont_time_lyapunov_app}.
We approximate this integral in the domain $[0, L]$  using the left Riemann sum. 
We discretize the interval using $\frac{L}{\Delta}$, which is assumed to be an integer, and intervals of size $\Delta>0$:
\begin{align}
\label{eq:riemann_sum_app}
    X \approx \Delta \sum_{k=0}^{\frac{L}{\Delta} - 1} e^{k \Delta A} B e^{k \Delta A^\dag} .
\end{align}
Suppose that we have a quantum state-preparation circuit for $B$ and a block-encoding circuit for $e^{\Delta A}$. We can implement the CP map $\mathcal{E} (\cdot) =  e^{\Delta A} \,\cdot \, e^{\Delta A^\dag}$ by executing the block-encoding circuit, measuring the ancilla, and postselecting. 
For the coefficients, we want to use Eq.~\eqref{eq:constant_coefficients_app} and $T = \frac{L}{\Delta} - 1$ so that the expected output state is
\begin{align}
\label{eq:master_state}
     \mathbb{E}[\rho] = \frac{\sum_{k=0}^{\frac{L}{\Delta} - 1} e^{k\Delta A} B e^{k\Delta A^\dag} }{\tr (\sum_{l=0}^{\frac{L}{\Delta} - 1} e^{l\Delta A} B e^{l\Delta A^\dag})} ,
\end{align}
which is the normalized version of Eq.~\eqref{eq:riemann_sum_app}. 
Let us introduce some notation for the calculations in this section. We define 
\begin{align}
    \mathcal{I}(L) \coloneq \int_0^L e^{tA} B e^{t A^\dag} dt
    \qquad \text{and} \qquad
    \mathcal{S}(L) \coloneq \Delta \sum_{k=0}^{\frac{L}{\Delta} -1} e^{k \Delta A} B e^{k \Delta A^\dag} .
\label{eq:IN_and_SN_definitions}
\end{align}
We are now ready to prove the theorem: we find the value of the parameters $L$ and $\Delta$ such that the expected output state is $\epsilon$-close to the normalized solution in trace distance. Since $B$ is positive semidefinite, we have $B = C C^\dag$ for some $C$. Then the integrand can be written in the form $(e^{tA} C)(e^{tA}C)^\dag$ and it is also positive semidefinite (similarly for the summation). It follows that both $\mathcal{I}(L)$ and $\mathcal{S}(L)$ are positive semidefinite, and, in turn, we have 
\begin{align}
\label{eq:property_a_alt}
    \tracenorm*{ \mathcal{I}(L) } = \tr[\mathcal{I}(L)] \qquad \text{and} \qquad  \tracenorm*{ \mathcal{S}(L) } = \tr[\mathcal{S}(L)].
\end{align}
Furthermore we have
\begin{align}
    \mathcal{I}(T) &\succeq \mathcal{I}(L) \text{ for } T\geq L,\\
    \mathcal{S}(T) &\succeq \mathcal{S}(L) \text{ for } T\geq L.
\end{align}

\begin{proof}[Proof of Theorem~\ref{thm:continuous_lyapunov}]
Let us denote $R\coloneq\max_k \Set{\Re( \lambda_k)}<0$.
Using the triangle inequality 
\begin{align}
    \frac{1}{2}\tracenorm*{\frac{X}{\tr(X)} - \mathbb{E}[\rho] }
    &=
    \frac{1}{2}\tracenorm*{\frac{ \mathcal{I}(\infty) }{\tr[\mathcal{I}(\infty)]} - \frac{\mathcal{S}(L)}{\tr[\mathcal{S}(L)]} } \\
    &\leq
    \underbrace{
    \frac{1}{2}\tracenorm*{\frac{ \mathcal{I}(\infty) }{\tr[\mathcal{I}(\infty)]} - \frac{\mathcal{S}(\infty)}{\tr[\mathcal{S}(\infty)]}}}_{\le\epsilon_1}
        +
    \underbrace{\frac{1}{2}\tracenorm*{\frac{ \mathcal{S}(\infty) }{\tr[\mathcal{S}(\infty)]} - \frac{\mathcal{S}(L)}{\tr[\mathcal{S}(L)]} }}_{\le\epsilon_2}.
    \label{eq:cont_tri_ineq}
\end{align}
The values $0 < \epsilon_1, \epsilon_2 < 1$ above are our desired bounds for these two distances.
The latter distance is bounded as in the discrete case, namely, as in the proof of Theorem~\ref{thm:discrete_lyapunov}. In particular, analogously to Eq.~\eqref{eq:upperbound_discrete}, we get
\begin{align}
    \frac{1}{2}\tracenorm*{\frac{ \mathcal{S}(\infty) }{\tr[\mathcal{S}(\infty)]} - \frac{\mathcal{S}(L)}{\tr[\mathcal{S}(L)]} }
    \le
    \opnorm{e^{\Delta A}}^{2[(L/\Delta)-1+1]}
    = e^{2RL}.
\end{align}
For this distance to be at most $\epsilon_2$, we can choose
\begin{align}
\label{eq:L_bound}
    L \geq \frac{\ln(1/\epsilon_2)}{- 2\max_j \Set{ \Re(\lambda_j) }} .
\end{align}
For the former distance, we apply Claim~\ref{clm:normalized_semidef_dist}, getting
\begin{align}
    \frac{1}{2}\tracenorm*{\frac{ \mathcal{I}(\infty) }{\tr[\mathcal{I}(\infty)]} - \frac{\mathcal{S}(\infty)}{\tr[\mathcal{S}(\infty)]} }
 \leq
    \frac{\tracenorm*{\mathcal{S}(\infty) - \mathcal{I}(\infty)}}{\tr[\mathcal{S}(\infty)]}.
\end{align}

\paragraph{Bounding numerator.}
We define the function $F(t)\coloneq e^{tA} B e^{tA^\dag}$ with derivative $F'(t)=e^{tA}A B e^{t A^\dag} + e^{tA}B A^\dag e^{t A^\dag}$.
Note that, for $t\ge \Delta l$, we have $F(t)=F(\Delta l)+\int_{\Delta l}^t F'(s)ds$.
We find
\begin{align}
    \tracenorm*{\mathcal{S}(\infty)-\mathcal{I}(\infty)}
    &=
    \tracenorm*{\Delta \sum_{l=0}^{\infty} e^{\Delta l A} B e^{\Delta l A^\dag} - \int_0^\infty e^{tA} B e^{t A^\dag} dt}\\
    &=
    \tracenorm*{\sum_{l=0}^{\infty} \int_{\Delta l}^{\Delta(l+1)} \Big( F(\Delta l) -  F(t) \Big) dt}\\
    &\leq
    \sum_{l=0}^{\infty} \int_{\Delta l}^{\Delta(l+1)} \Big(\int_{\Delta l}^t \tracenorm*{F'(s)} ds \Big) dt.
\end{align}
We note that
\begin{align}
    \tracenorm*{F'(s)}
    & \le
    \opnorm*{e^{s A}} \tracenorm*{A B + B A^\dag} \opnorm*{e^{s A^\dag}}\\
    & =
    e^{2s\max_k \Set{ \Re( \lambda_k)}} \tracenorm*{A B + B A^\dag}\\
    & \le
    2 \opnorm{A} e^{2Rs}, \label{eq:vn_trace_ineq_step}
\end{align}
since $\tracenorm{XYZ}\leq \tracenorm{X Y}\opnorm{Z}\leq
\opnorm{X}\tracenorm{Y}\opnorm{Z}$ by two applications of a corollary of von Neumann's trace inequality (follows from Ref.~\cite[Theorem
7.4.1.3(d)]{hornMatrixAnalysis2017}). In the last step we also use von Neumann's trace inequality and $\tracenorm*{B}=1$. We therefore get
\begin{align}
    \tracenorm*{\mathcal{S}(\infty)-\mathcal{I}(\infty)}
    & \le
    2 \opnorm{A} \sum_{l=0}^{\infty} \int_{\Delta l}^{\Delta(l+1)} \Big(\int_{\Delta l}^t e^{2Rs} ds \Big) dt
    \\
    & \le
    2 \opnorm{A} \sum_{l=0}^{\infty} e^{2R\Delta l}\int_{\Delta l}^{\Delta(l+1)} \left(t-\Delta l \right) dt
    \\
    & =
    \opnorm{A}\Delta^2 \sum_{l=0}^{\infty} e^{2R\Delta l}
    \\
    & =
    \frac{\opnorm{A}\Delta^2}{1-e^{2R\Delta}}.
\end{align}

\paragraph{Bounding denominator.}
We have
\begin{align}
    \tr[\mathcal{S}(\infty)]
    &=
    \Delta \sum_j B_{jj} \sum_{l=0}^{\infty} 
    e^{2\Delta l\Re(\lambda_j)}
    =
    \sum_j \frac{B_{jj}\Delta}{1-e^{2\Delta\Re(\lambda_j)}} 
    \\
    & \ge
    \min_j \Set*{ \frac{\Delta}{1-e^{2\Delta\Re(\lambda_j)}} }
    =
    \frac{\Delta}{1-\min_j \Set*{ e^{2\Delta\Re(\lambda_j)} }}
    =
    \frac{\Delta}{1-e^{2r\Delta}} ,
\end{align}
where we use $\tr(B)=1$ and introduce the notation $r\coloneq \min_j\Set{\Re(\lambda_j)}\le R < 0$.
Note that $\frac{1-e^y}{-y}$ is a positive strictly increasing function in $y$, which means that $\frac{1-e^{2r\Delta}}{-2r\Delta}\le \frac{1-e^{2R\Delta}}{-2R\Delta}$.
By combining the bound for the numerator and denominator, we get that 
\begin{align}
    \frac{1}{2}\tracenorm*{\frac{ \mathcal{I}(\infty) }{\tr[\mathcal{I}(\infty)]} - \frac{\mathcal{S}(\infty)}{\tr[\mathcal{S}(\infty)]} }
    \le \opnorm{A}\Delta \frac{1-e^{2r\Delta}}{1-e^{2R\Delta}}
    \le \opnorm{A}\Delta \frac{r}{R}.
\end{align}
For this distance to be at most $\epsilon_1$, we can choose
\begin{align}\label{eq:Delta_bound}
    \Delta \leq \frac{\epsilon_1}{\opnorm{A}}\frac{\max_j\Set{\Re(\lambda_j)}}{\min_k\Set{\Re(\lambda_k)}}.
\end{align}
Recall that $T=\frac{L}{\Delta}-1$. Combining our choices for $L$, Eq.~\eqref{eq:L_bound}, and $\Delta$, Eq.~\eqref{eq:Delta_bound}, we arrive at
\begin{align}\label{eq:continuous_lyapunov_time_bound}
    T
    &\geq
    \ceil*{ \frac{1}{2\epsilon_1} \ln( \frac{1}{\epsilon_2}) \frac{ \opnorm{A} \abs{\min_k \Set{ \Re(\lambda_k) }}}{ \abs{\max_j \Set{ \Re(\lambda_j) }}^2 } } .
\end{align}
Note that $\abs{\min_k \Set{ \Re(\lambda_k) }}\le\opnorm{A}$, with equality for Hermitian $A$.
In our use context, the factor $\opnorm{A} / \abs{\max_j\Set{\Re(\lambda_j)}}$ is similar to the condition number $\kappa$. Indeed if $A$ is Hermitian, and recalling that all its eigenvalues are negative, then $\abs{\max_j\Set{\Re(\lambda_j)}} = \min_j \Set{\abs{\lambda_j}}$, so that $T \geq \ceil*{ \frac{\kappa^2}{2\epsilon_1}\ln(\frac{1} {\epsilon_2}) }$.

\end{proof}

\subsection{Robustness}
\label{app:robustness}

In this section we discuss the effect of approximations when implementing the algorithm.
Conceptually, it is straightforward to include such error in the analysis using the triangle inequality. Let $\sim$ mark approximate quantities. For the discrete Lyapunov equation, Eq.~\eqref{eq:derivation_2} is replaced by 
\begin{align}
    \frac{1}{2} \tracenorm*{ \frac{X}{\tr(X)} - \mathbb{E}[\rho] }
    =
    \underbrace{\frac{1}{2}\tracenorm*{ \frac{ \mathcal{S}(\infty) }{\tr[\mathcal{S}(\infty)]} - \frac{\mathcal{S}(T)}{\tr[\mathcal{S}(T)]} }}_{\leq \epsilon} + \underbrace{\frac{1}{2}\tracenorm*{ \frac{ \mathcal{S}(T) }{\tr[\mathcal{S}(T)]} - \frac{\tilde{\mathcal{S}}(T)}{\tr[\tilde{\mathcal{S}}(T)]} }}_{\leq\epsilon_{\sim}},
\end{align}
where $T$ is still given by Eq.~\eqref{eq:lyapunov_time_and_error}, and we introduce a parameter $\epsilon_{\sim} \geq 0$ to bound the implementation error. 
In the continuous case, Eq.~\eqref{eq:cont_tri_ineq} is replaced by
\begin{align}
    \frac{1}{2}\tracenorm*{\frac{X}{\tr(X)} - \mathbb{E}[\rho] }
    \leq
    \underbrace{
    \frac{1}{2}\tracenorm*{\frac{ \mathcal{I}(\infty) }{\tr[\mathcal{I}(\infty)]} - \frac{\mathcal{S}(\infty)}{\tr[\mathcal{S}(\infty)]}}}_{\le\epsilon_1}
        +
    \underbrace{\frac{1}{2}\tracenorm*{\frac{ \mathcal{S}(\infty) }{\tr[\mathcal{S}(\infty)]} - \frac{\mathcal{S}(L)}{\tr[\mathcal{S}(L)]} }}_{\le\epsilon_2}
        +
    \underbrace{\frac{1}{2}\tracenorm*{ \frac{ \mathcal{S}(L) }{\tr[\mathcal{S}(L)]} - \frac{\tilde{\mathcal{S}}(L)}{\tr[\tilde{\mathcal{S}}(L)]} }}_{\leq \epsilon_{\sim}},
\end{align}
and $L=(T+1)\Delta$ is still given by the parameters in Theorem~\ref{thm:continuous_lyapunov}.

In the remainder of the section we inspect the newly added error term in order to obtain a more practical expression. We begin by relating the error term to the approximate implementation $\tilde{\mathcal{E}}$ of CP map $\mathcal{E}$. For the discrete-time Lyapunov equation we use Claim~\ref{clm:normalized_semidef_dist} to obtain
\begin{align}\label{eq:robustness_bound_1}
    \frac{1}{2}\tracenorm*{\frac{ \mathcal{S}(T) }{\tr[\mathcal{S}(T)]} - \frac{\tilde{\mathcal{S}}(T)}{\tr[\tilde{\mathcal{S}}(T)]} }
    \leq
    \frac{\tracenorm*{\mathcal{S}(T) - \tilde{\mathcal{S}}(T)}}{\tr[\mathcal{S}(T)]}
    =
    \frac{\tracenorm*{\sum_{k=0}^{T} \mathcal{E}^k(B) - \tilde{\mathcal{E}}^k(B)}}{\sum_{l =0}^{T}\tr[\mathcal{E}^{l}(B)]} \leq \sum_{k=0}^{T} \tracenorm*{\mathcal{E}^k(B) - \tilde{\mathcal{E}}^k(B)}.
\end{align}
For the denominator, we use that $\sum_{l=0}^{T}\tr[\mathcal{E}^{l}(B)] \geq \tr[\mathcal{E}^{0}(B)] = \tr(B) = 1$. In the continuous-time case the numerator and denominator of Eq.~\eqref{eq:robustness_bound_1} acquire a factor $\Delta$, which cancels, and we can also plug $T = L/\Delta -1$ to obtain an identical expression. Thus the following steps are valid for both continuous- and discrete-time cases.

Let us now consider the diamond norm of a linear map $\mathcal{A} : \mathbb{C}^{d\times d} \rightarrow \mathbb{C}^{d\times d}$, which is defined as $\diamondnorm{\mathcal{A}} = \max \Set{ \tracenorm{(\mathcal{A} \otimes I_d) (X) } \given X \in \mathbb{C}^{2d\times 2d}, \tracenorm{X} \leq 1 }$. The diamond norm is submultiplicative under composition of maps, $\diamondnorm{\mathcal{B} \circ \mathcal{A} } \leq \diamondnorm{ \mathcal{B}} \diamondnorm{\mathcal{A}}$, thus 
\begin{align}
    \diamondnorm{ \mathcal{B} \circ \mathcal{A} - \mathcal{D} \circ \mathcal{C} } &\leq \diamondnorm{\mathcal{B} \circ \mathcal{A} - \mathcal{B} \circ \mathcal{C} } + \diamondnorm{\mathcal{B} \circ  \mathcal{C} - \mathcal{D} \circ \mathcal{C} } \\
    & \leq \diamondnorm{\mathcal{B}}  \diamondnorm{\mathcal{A} - \mathcal{C} }  + \diamondnorm{\mathcal{B} - \mathcal{D}} \diamondnorm{\mathcal{C}} . \label{eq:subadd_diamond}
\end{align}
For trace nonincreasing maps, \eg\ $\tracenorm{ (\mathcal{B} \otimes I_d) (X)} \leq \tracenorm{X} \leq 1$, we have that $\diamondnorm{\mathcal{B}} \leq 1$. So Eq.~\eqref{eq:subadd_diamond} indicates subadditivity with respect to composition of trace nonincreasing maps. Using this fact recursively in Eq.~\eqref{eq:robustness_bound_1} we have
\begin{align}
    \sum_{k=0}^{T} \tracenorm*{\mathcal{E}^k(B) - \tilde{\mathcal{E}}^k(B)}
    \leq \sum_{k=0}^{T} \diamondnorm*{\mathcal{E}^k - \tilde{\mathcal{E}}^k}
    \leq \sum_{k=0}^{T} k \diamondnorm*{\mathcal{E} - \tilde{\mathcal{E}}}
    \leq \frac{T(T+1)}{2} \diamondnorm*{\mathcal{E} - \tilde{\mathcal{E}}}.\label{eq:robustness_bound_2}
\end{align}
In the second inequality we use $\tracenorm*{\mathcal{F}(X)}
\leq \diamondnorm*{\mathcal{F}}\tracenorm*{X}$.
Recall that our algorithm utilizes an approximate block encoding of the Kraus operator $M$ to produce an approximate implementation of the map $\mathcal{E}(\cdot) = M \,\cdot \, M^\dag$. Let us recap how this works. 

Suppose we have a block encoding $U_{\tilde{M}}$ such that $\tilde{M} = (\bra{G} \otimes I^{\otimes n}) U_{\tilde{M}} (\ket{G} \otimes I^{\otimes n})$ for a known state $\ket{G}$. Note the absence of subnormalization constants as these could bias the transition probabilities of our algorithm (in Appendix~\ref{app:block_encoding_exp} we show how to produce such block encoding in practice). We now attach the ancilla register $\dyad{G}$ to the system register $\rho$ and apply $U_{\tilde{M}}$, producing $U_{\tilde{M}}( \dyad{G} \otimes \rho)U_{\tilde{M}}^\dag$. We then perform a projective measurement on the ancilla register with elements $\Set{\dyad{G}, I - \dyad{G}}$. If the measurement outcome is $1$, we obtain $\dyad{G} \otimes \frac{\mathcal{\tilde{E}}(\rho)}{\tr[\mathcal{\tilde{E}}(\rho)]}$. Discarding the ancilla register, we obtain the desired approximation and the algorithm enters the `continue' branch in Eq.~\eqref{eq:intuition}. If the measurement outcome is $0$, the algorithm enters the `restart' branch. So how does the error in $M$ relates to the error in $\mathcal{E}$? We have that
\begin{align}
    \diamondnorm{\mathcal{E} - \tilde{\mathcal{E}}}
    &=
    \max_\rho \tracenorm{(\mathcal{E} \otimes I_d)(\rho) - (\tilde{\mathcal{E}} \otimes I_d)(\rho)}
    \\
    &=
    \max_\rho \tracenorm{ (M \otimes I_d) \rho (M \otimes I_d)^\dag - (\tilde{M} \otimes I_d) \rho (\tilde{M} \otimes I_d)^\dag}
    \\
    &=
    \max_\rho \tracenorm{ (M - \tilde{M} \otimes I_d) \rho (M \otimes I_d)^\dag + (M \otimes I_d) \rho (M - \tilde{M} \otimes I_d)^\dag - (M - \tilde{M} \otimes I_d) \rho (M - \tilde{M} \otimes I_d)^\dag}
    \\
    &\leq
    \max_\rho \left\{ \tracenorm{ (M - \tilde{M} \otimes I_d) \rho (M \otimes I_d)^\dag} + \tracenorm{(M \otimes I_d) \rho (M - \tilde{M} \otimes I_d)^\dag} + \tracenorm{(M - \tilde{M} \otimes I_d) \rho (M - \tilde{M} \otimes I_d)^\dag} \right\}
    \\
    &\leq
    \max_\rho \left\{ 2\opnorm*{M - \tilde{M} \otimes I_d } \tracenorm{\rho} \opnorm*{M \otimes I_d} + \opnorm*{M - \tilde{M} \otimes I_d}^2 \tracenorm{\rho} \right\}
    \\
    &\leq
    2\opnorm*{ M - \tilde{M} \otimes I_d } \opnorm*{M \otimes I_d} + \opnorm*{M - \tilde{M} \otimes I_d}^2
    \\
    &\leq
    2\opnorm*{ M - \tilde{M} } + \opnorm*{M - \tilde{M}}^2 .\label{eq:robustness_bound_3}
\end{align}
Here we use $\tracenorm{XYZ} \leq
\opnorm{X}\tracenorm{Y}\opnorm{Z}$ (see comment after Eq.~\eqref{eq:vn_trace_ineq_step}), $\opnorm*{X^\dag} = \opnorm{X}$, $\tracenorm{\rho}=1$ and $\opnorm*{X \otimes I_d}=\opnorm{X}$. In the final step we also use that $\opnorm{M} \leq 1$ because $M$ is the Kraus operator of a trace nonincreasing map. 

Putting together Eqs.~\eqref{eq:robustness_bound_1},~\eqref{eq:robustness_bound_2}, and~\eqref{eq:robustness_bound_3} we arrive at
\begin{align}
    \frac{1}{2}\tracenorm*{\frac{ \mathcal{S}(T) }{\tr[\mathcal{S}(T)]} - \frac{\tilde{\mathcal{S}}(T)}{\tr[\tilde{\mathcal{S}}(T)]} }
    \leq T(T+1)\left(\opnorm*{ M - \tilde{M} } + \frac{1}{2}\opnorm*{M - \tilde{M}}^2 \right) \leq \epsilon_\sim.
\end{align}
In the trivial case $T=0$ the algorithm never uses the CP map so the block-encoding error is irrelevant. When $T\neq 0$ instead, solving for the block-encoding error we see that $\opnorm*{ M - \tilde{M} } \leq \sqrt{ 1 + \frac{2 \epsilon_{\sim}}{T(T+1)}} - 1$ with $T$ given in Theorem~\ref{thm:discrete_lyapunov} for the discrete case and Theorem~\ref{thm:continuous_lyapunov} for the continuous case. This expression is rather complicated, so we make one further approximation. Noting that $x/4 \leq \sqrt{1 + x} - 1$ for $0 \leq x \leq 8$, we conclude that
\begin{align}\label{eq:robustness_bound_4}
    \opnorm*{ M - \tilde{M} } \leq \frac{\epsilon_\sim}{2T(T+1)} 
    \quad \implies \quad \frac{1}{2}\tracenorm*{\frac{ \mathcal{S}(T) }{\tr[\mathcal{S}(T)]} - \frac{\tilde{\mathcal{S}}(T)}{\tr[\tilde{\mathcal{S}}(T)]} }
    \leq
    \epsilon_\sim.
\end{align}

\subsection{Block encoding the Kraus operator for the continuous-time Lyapunov equation}
\label{app:block_encoding_exp}

In the previous section we showed how the implementation error is related to the block-encoding error of the Kraus operator $M$. Here we provide an explicit construction for the continuous-time Lyapunov equation, where we specifically need $M = e^{\Delta A}$ (matrix inversion II is solved via a reduction to the continuous-time Lyapunov equation and thus the  construction is very similar). For simplicity, we assume $A \prec 0$, $\Delta >0$, and that we are given a block encoding of $-A/\opnorm{A}$. That is, we have access to a circuit $U_A$ such that $-A/\opnorm{A} = (\bra{G} \otimes I^{\otimes n}) U_A (\ket{G} \otimes I^{\otimes n})$ for some state $\ket{G}$. We now construct the block encoding of $M$ using the quantum singular value transformation (QSVT)~\cite{gilyenQuantumSingularValue2019}.

QSVT is a technique to apply a polynomial to a block-encoded matrix with query complexity proportional to the degree of the polynomial. In its basic version, the polynomial must have a definite parity (even or odd) and be bounded, $|P(x)| \leq 1$ for $x \in [-1,1]$. Our application calls for the use of a polynomial approximation to $e^{-\beta x}$. This is discussed in Ref.~\cite{gilyenQuantumSingularValue2019} where, however, a subnormalization factor is used to ensure that the polynomial remains bounded. The subnormalization factor is related to the probability of success of the block encoding upon projective measurement of the ancilla register. The subnormalization factor can therefore bias the carefully designed transition probabilities in our algorithm. In the QSVT literature, the subnormalization factor is often removed using the amplitude amplification algorithm, but this requires quantum circuits that perform reflections with respect to the initial state. Due to the probabilistic nature of our algorithm, the system register is in an unknown mixed state at all times, and thus we cannot construct such reflections. In the following we present a solution that block encodes our operator $M$ without any subnormalization. 

Let us define $H \coloneq -A/\opnorm{A} \succ 0$ and $\beta \coloneq \Delta \opnorm{A}$, so that $e^{-\beta \abs{H}} = e^{\Delta A} = M$. The condition number is $\kappa \coloneq \kappa(H) = \kappa(A)$. We now seek a polynomial approximation to the function $f(x)=e^{-\beta \abs{x}}$, which is accurate in the interval $[\frac{1}{\kappa}, 1]$ where the spectrum of $H$ lies. We resort to the following result.

\begin{theorem}[Restatement of Theorem 68 in Ref.~\cite{gilyenQuantumSingularValue2019}, bounded polynomial approximation based on multiple local Taylor series]
\label{thm:poly_approx}
Let $J\in \mathbb{N}$, $(x_j,r_j,\mu_j)\in[-1,1]^J\times(0,2]^J\times(0,1]^J$, such that $x_j\colon j\in [J]$ is monotone increasing, and $\mu_j\leq r_j$ for all $j\in[J]$. Let $I\coloneq\bigcup_{j\in [J]}[x_j-r_j,x_j+r_j]$ be the union of the intervals $[x_j-r_j,x_j+r_j]$, and suppose that for all $i<j\in[J]$ such that $j-i\geq 2$ we have that $x_i+r_i < x_j-r_j$. 
Let $\mu=\min\left[\min_{j\in [J]} \mu_j,\min_{j\in [J-1]} |x_{j+1} - r_{j+1} - x_{j} - r_j |\right]$. Let $f:I+[-\frac{\mu}2,\frac{\mu}2]\rightarrow \mathbb{C}$, $B\in \mathbb{R}_+$ be such that for all $j\in[J]$ we have $f(x_j+x)=\sum_{k=0}^{\infty}a^{(j)}_k x^k$ for all $x\in [-r_j-\frac{\mu_j}2, r_j+\frac{\mu_j}2]$ 
and $\sum_{k=0}^{\infty}(r_j+\mu_j)^k |a^{(j)}_k|\leq B$. Let $\delta \in\!\left(0,\frac{1}{2BJ}\right]$, then there is an efficiently computable polynomial $P\in \mathbb{C}[x]$ of degree $\mathcal{O}[\frac{J}{\mu}\log\left(\frac{BJ}{\delta}\right)]$ such that
\begin{align}
\max_{x \in I} |f(x)- P(x)| &\leq \delta,\\
\max_{x \in [-1,1]} |P(x)| &\leq \max_{x \in I+[-\mu/2,\mu/2]} |f(x)|,\\
\max_{x \in [-1,1]\setminus \left(I+[-\mu/2,\mu/2]\right)} |P(x)| &\leq \delta.
\end{align}
\end{theorem}
Let us choose $J=2$, and set $(x_1, r_1, \mu_1) = (-1, 1 - \frac{1}{\kappa}, \frac{1}{2\kappa})$ and $(x_2, r_2, \mu_2) = (1, 1 - \frac{1}{\kappa}, \frac{1}{2\kappa})$. First, we Taylor expand the functions $f_1\colon (-\infty, 0) \rightarrow \mathbb{R}, x\mapsto e^{\beta x}$ around $x_1=-1$ and $f_2\colon (0, \infty) \rightarrow \mathbb{R}, x\mapsto e^{-\beta x}$ around $x_2=1$ (both functions are analytic on their respective open domains):
\begin{align}
    f_1(x)
    =
    \sum_{\ell=0}^\infty \frac{f_1^{(\ell)}(-1)}{\ell!} (x + 1)^\ell
    =
    \sum_{\ell=0}^\infty \frac{\beta^\ell e^{-\beta}}{\ell!} (x + 1)^\ell 
    &\quad\Leftrightarrow\quad
    f_1(-1+x)
    =
    \sum_{\ell=0}^\infty \frac{\beta^\ell e^{-\beta}}{\ell!} x^\ell \quad\text{for } x\in (-\infty, 1),\label{eq:f_1}\\
    f_2(x)
    =
    \sum_{\ell=0}^\infty \frac{f_2^{(\ell)}(1)}{\ell!} (x - 1)^\ell
    =
    \sum_{\ell=0}^\infty \frac{(-\beta)^\ell e^{-\beta}}{\ell!} (x - 1)^\ell &\quad\Leftrightarrow\quad f_2(1+x)
    =
    \sum_{\ell=0}^\infty \frac{(-\beta)^\ell e^{-\beta}}{\ell!} x^\ell \quad\text{for } x\in (-1, \infty).\label{eq:f_2}
\end{align}
Now we consider the intervals $[x_j - r_j - \frac{\mu_j}{2}, x_j + r_j + \frac{\mu_j}{2}]$ for $j=1, 2$:
\begin{align}
    I_1 &= \left[-2 +\frac{3}{4\kappa}, -\frac{3}{4\kappa}\right], \\
    I_2 &= \left[\frac{3}{4\kappa}, 2-\frac{3}{4\kappa} \right].
\end{align}
Our target function $f(x)=e^{-\beta\abs{x}}$ matches $f_1(x)$ for $x\in I_1 \subset (-\infty, 0)$ and $f_2(x)$ for $x\in I_2 \subset (0, \infty)$. Hence, we can express the target function using the expansions in Eqs.~\eqref{eq:f_1} and~\eqref{eq:f_2} on the intervals $I_1$ and $I_2$, respectively, as required by the theorem:
\begin{align}
    f(-1+x) = \sum_{\ell=0}^\infty \frac{\beta^\ell e^{-\beta}}{\ell!} x^\ell &\quad\text{for } x\in \left[-1+\frac{3}{4\kappa}, 1-\frac{3}{4\kappa}\right] \subset (-\infty, 1),\\
    f(1+x) = \sum_{\ell=0}^\infty \frac{(-\beta)^\ell e^{-\beta}}{\ell!} x^\ell &\quad\text{for } x\in \left[-1+\frac{3}{4\kappa}, 1-\frac{3}{4\kappa}\right] \subset (-1, \infty).
\end{align}
This defines the coefficients $a_\ell^{(1)} = \frac{\beta^\ell e^{-\beta}}{\ell!}$ and  $a_\ell^{(2)} = \frac{(-\beta)^\ell e^{-\beta}}{\ell!}$ for all $\ell\geq 0$, and the other parameters
\begin{align}
    &\mu = \frac{1}{2\kappa},\\
    &I = \left[-2 +\frac{1}{\kappa}, -\frac{1}{\kappa}\right] \cup \left[\frac{1}{\kappa}, 2-\frac{1}{\kappa} \right],\\
    &B= \sum_{\ell=0}^\infty \left(1 - \frac{1}{\kappa} + \frac{1}{2\kappa} \right)^\ell \left|\frac{(\mp\beta)^\ell  e^{-\beta}}{\ell!} \right| = e^{-\beta} \sum_{\ell=0}^\infty \frac{\beta^\ell(1 - \frac{1}{2 \kappa})^\ell}{\ell!} = e^{- \frac{\beta}{2\kappa}}. 
\end{align}
By the theorem, for $\delta\in(0, \frac{1}{4B}]$ there exists an efficiently computable polynomial such that
\begin{align}
    \max_{x \in I} \left| e^{-\beta |x|} - P(x) \right| \leq \delta,
    \qquad\qquad 
    \max_{x \in [-1, 1]} |P(x)| \leq e^{-\frac{3\beta}{4\kappa}} \leq 1.
\end{align}
Since $P$ is complex and not guaranteed to have a definite parity, it cannot be directly implemented by QSVT. Let us then take the real, even part of the polynomial, $Q(x) = [P(x) + P(-x) + P(x)^* + P(-x)^*]/4$. Noting that the target function is real and even, we can also write $f(x) = [f(x) + f(-x) + f(x)^* + f(-x)^*]/4$. Simple applications of the triangle inequality show that $\max_{x\in I} |f(x) - Q(x)| \leq \delta$, and that $\max_{x \in [-1, 1]} |Q(x)| \leq 1$. This polynomial can be implemented by QSVT using Theorem 17 and Corollary 10 in Ref.~\cite{gilyenQuantumSingularValue2019}. More specifically, we can construct a block encoding $U_{\tilde{M}}$ such that $\tilde{M} = Q(- A/\opnorm{A}) = (\bra{G} \otimes I^{\otimes n}) U_{\tilde{M}} (\ket{G} \otimes I^{\otimes n})$ with the same state $\ket{G}$ used in the block encoding $U_A$. By Theorem~\ref{thm:poly_approx} the query complexity, \ie\ the number of uses of $U_A$, is proportional to the degree of the polynomial $Q$. That is, $\mathcal{O}\left[\frac{J}{\mu}\log(\frac{BJ}{\delta})\right] = \mathcal{O}\left[\kappa \log (\frac{1}{\delta})\right]$.

Specifically choosing $\delta = \frac{\epsilon_\sim}{2T(T+1)}$ we arrive at $\opnorm{M - \tilde{M}} \leq \frac{\epsilon_\sim}{2T(T+1)}$ as required by Eq.~\eqref{eq:robustness_bound_4}. The total number of queries to $U_A$ is then $\mathcal{O}\left[\kappa \log (\frac{T}{\epsilon_\sim})\right]$. Note, however, that Theorem~\ref{thm:poly_approx} requires $\delta \leq \frac{1}{4B}$. Plugging our parameters in we find that the construction is valid for $\epsilon_\sim \leq \frac{T(T+1)}{2}e^{\frac{\Delta\opnorm{A}}{2\kappa}}$.

\section{Lower bound for the discrete-time algorithm}
\label{app:optimality_of_Tstar}

We can show that the choice of $T^* \coloneq  \ceil*{ \frac{1}{2}\ln (\frac{1}{\epsilon}) \big/ \ln (\frac{1}{\opnorm{A}}) }$ in Theorem~\ref{thm:discrete_lyapunov} is
asymptotically optimal. In particular, the following theorem implies that, if we had chosen $T^*< \frac{1}{4}\ln (\frac{1}{\epsilon}) \big/ \ln (\frac{10}{\opnorm{A}})$ in Theorem~\ref{thm:discrete_lyapunov}, then it would have been false.

\begin{theorem}
\label{thm:hardness_proof_discrete_time}
    For every $1/2\le \lambda<1$, there exists a normal matrix $A$ of operator
    norm $\lambda$ and a state $B$ such that the matrix
    $\mathcal{S}(T)$ defined in Eq.~\eqref{eq:sol_disc_time_lyapunov} satisfies
    \begin{align}
        \frac{1}{2}\tracenorm*{\frac{\mathcal{S}(\infty)}{\tr[\mathcal{S}(\infty)]}-\frac{\mathcal{S}(T)}{\tr[\mathcal{S}(T)]}}
        \ge \frac{\lambda^{4T}}{10}.
    \end{align}
\end{theorem}

Note that Theorem~\ref{thm:hardness_proof_discrete_time} does not rule out the existence of an algorithm for solving the discrete-time Lyapunov equation that is more efficient than Algorithm~\ref{alg:cap}. That is, an algorithm that produces a state $\rho$ satisfying $\frac{1}{2} \tracenorm*{\frac{X}{\tr(X)} - \rho} \leq \epsilon$ while consuming less copies of the initial state $\rho_0=B$ and with less calls to the block encoding of $A$.

\begin{proof}[Proof of Theorem~\ref{thm:hardness_proof_discrete_time}]
We choose the maximally mixed state for a single qubit, $B = \frac{1}{2}(\dyad{0}+\dyad{1})$, and
\begin{align}
    A \coloneq \sqrt\mu\dyad{0}
    +\sqrt\nu\dyad{1} ,
\end{align}
with $0<\nu<\mu<1$ and $\mu=\lambda^2$. Note that the operator norm of $A$ is $\lambda$, as required. Eventually, after 
Eq.~\eqref{eq:rhoT_rhoInf_lower}, we will effectively choose $\nu\coloneq (3\lambda^2-1)/2$.
For $A$ and $B$ chosen as above, from Eq.~\eqref{eq:sol_disc_time_lyapunov} we see that 
\begin{align}
    \mathcal{S}(T)
    =
    \frac{1}{2} \sum_{k=0}^{T} \left(\mu^k\dyad{0}
    +\nu^k\dyad{1}\right)
    =
    \frac{1}{2}  \left(\frac{1-\mu^{T+1}}{1-\mu}\dyad{0}
    +\frac{1-\nu^{T+1}}{1-\nu}\dyad{1}\right).
\end{align}
For brevity, define $\eta(T)\coloneq\mathcal{S}(T)/\tr[\mathcal{S}(T)]$, so we have
\begin{align}
    \eta(T)  
    = 
    \frac{
    \frac{1-\mu^{T+1}}{1-\mu}\dyad{0}
    +\frac{1-\nu^{T+1}}{1-\nu}\dyad{1}
    }{
    \frac{1-\mu^{T+1}}{1-\mu}
    +\frac{1-\nu^{T+1}}{1-\nu}
    }
    = 
    \frac{\dyad{0}}{
    1+\frac{1-\mu}{1-\nu}\frac{1-\nu^{T+1}}{1-\mu^{T+1}}}
    +
    \frac{\dyad{1}}{
    1+\frac{1-\nu}{1-\mu}\frac{1-\mu^{T+1}}{1-\nu^{T+1}}}.
\end{align}
Because $\eta(T)-\eta(\infty)$ is an Hermitian matrix of trace $0$ and it is diagonal in the basis $\Set{\ket{0}, \ket{1}}$, its diagonal entries have the same absolute value yet opposite signs. Hence, by symmetry we have
\begin{align}
    \frac{\tracenorm*{\eta(T)-\eta(\infty)}}{2}
    & =
    \abs*{\innerp{0}{[\eta(T)-\eta(\infty)]}{0}}
    \\
    &
    = 
    \frac1{1+\frac{1-\mu}{1-\nu}}
    -\frac1{1+\frac{1-\mu}{1-\nu}\frac{1-\nu^{T+1}}{1-\mu^{T+1}}}
    \\
    & = 
    \frac{\frac{1-\mu}{1-\nu}\left(\frac{1-\nu^{T+1}}{1-\mu^{T+1}}-1\right)}{\left(1+\frac{1-\mu}{1-\nu}\right)\left(1+\frac{1-\mu}{1-\nu}\frac{1-\nu^{T+1}}{1-\mu^{T+1}}\right)}
    \\
    & = 
    \frac{(1-\mu)(1-\nu)(\mu^{T+1}-\nu^{T+1})}{(1-\nu+1-\mu)\left[(1-\nu)(1-\mu^{T+1})+(1-\mu)(1-\nu^{T+1})\right]}.
\end{align}
Let $0<\delta\le 1/2$ be such that $\mu=1-\delta$ and, in turn, $c>1$ be such that $\nu=1-c\delta$. Hence we have
\begin{align}
    \frac{\tracenorm*{\eta(T)-\eta(\infty)}}{2}
    & =  
    \frac{c\delta^2\left[(1-\delta)^{T+1}-(1-c\delta)^{T+1}\right]}{(c\delta+\delta)\left\{c\delta[1-(1-\delta)^{T+1}]+\delta[1-(1-c\delta)^{T+1}]\right\}}
    \\ & =  
    \label{eq:exact_rhoT_rhoInf}
    \frac{c\left[(1-\delta)^{T+1}-(1-c\delta)^{T+1}\right]}{(c+1)\left[c+1-c(1-\delta)^{T+1}-(1-c\delta)^{T+1}\right]}.
\end{align}
For bounding the numerator of Eq.~\eqref{eq:exact_rhoT_rhoInf}, we will use
\begin{align}
    (1-\delta)^{T+1}-(1-c\delta)^{T+1}
    & =
    [(1-\delta)-(1-c\delta)] \sum_{k=0}^{T}(1-\delta)^{k}(1-c\delta)^{T-k}
    \\
    & \ge
    (c-1)\delta\sum_{k=0}^{T}(1-c\delta)^{k}(1-c\delta)^{T-k}
    =
    (c-1)\delta (T+1) (1-c\delta)^{T}.
\end{align}
Above, the former equality is analogous to the equality $1-x^{T+1}=(1-x)\sum_{k=0}^{T}x^k$ concerning the finite geometric series where we have taken $x\coloneq\frac{1-c\delta}{1-\delta}$. 
For bounding the denominator of Eq.~\eqref{eq:exact_rhoT_rhoInf}, we will use
\begin{align}
    c(1-\delta)^{T+1}+(1-c\delta)^{T+1}
    & \ge c[1-\delta (T+1)]+[1-c\delta(T+1)]
    = c+1-2c\delta (T+1).
\end{align}
By incorporating the two inequalities, we get
\begin{align}
    \label{eq:rhoT_rhoInf_lower}
    \frac{\tracenorm*{\eta(T)-\eta(\infty)}}{2}
    \ge  
    \frac{c(c-1)\delta (T+1) (1-c\delta)^{T}}{(c+1)[c+1-c-1+2c\delta (T+1)]}
    = 
    \frac{(c-1)(1-c\delta)^{T}}{2(c+1)}.
\end{align}
We choose $c\coloneq\frac{3}{2}$, for which we have
$\nu=1-c\delta\ge(1-\delta)^2=\mu^2=\lambda^4$ because $\delta\le 1/2$. Thus, we
get
\begin{align}
    \frac{\tracenorm*{\eta(T)-\eta(\infty)}}{2}
    \ge   
    \frac{\left(\frac{3}{2}-1\right)\lambda^{4T}}{2\left(\frac{3}{2}+1\right)}
    =
    \frac{\lambda^{4T}}{10},
\end{align}
as claimed.
\end{proof}

\section{Estimating expectation values, overlaps, and normalization constants}\label{app:expectation_values}

Here we discuss how to read-off information from the output state of the quantum algorithm. Let us assume $\rho_{}$ is a state generated by the probabilistic algorithm, \eg\ encoding the solution to a Lyapunov equation, or a matrix inverse. It is straightforward to measure the expectation $\tr(\rho_{} O)$ of simple observables, such as Pauli matrices $O \in \Set{I, \sigma_X, \sigma_Y, \sigma_Z}^{\otimes n}$ or a linear combination thereof $O = \sum_i a_i O_i$ where $a_i \in \mathbb{R}$.

We can also interpret $\rho_{}$ itself as an observable. The task is then to estimate the expectation $\innerp{\psi}{\rho_{}}{\psi}$ with respect to another state $\ket{\psi}$. If we know the state-preparation circuit $U\ket{0} = \ket{\psi}$ then we can simply prepare $U^\dag \rho_{} U$ and measure the overlap with the zero state. If we do not know $U$, but have access to copies of $\ket{\psi}$ then we can use the well-known SWAP test as follows. We start the computation from $\dyad{+} \otimes \rho_{} \otimes \dyad{\psi}$, where the first register is an ancilla qubit in state $\ket{+}=\frac{1}{\sqrt{2}}(\ket{0} + \ket{1})$. We apply a circuit $C_S$ that swaps the second and third registers if the ancilla is in state $\ket{1}$. Finally, we measure the $\sigma_X \otimes I \otimes I$ observable. It can be verified that
\begin{align}
    \tr[ (\sigma_X \otimes I \otimes I) C_S (\dyad{+} \otimes \rho_{} \otimes \dyad{\psi}) C_S^\dag ] = \innerp{\psi}{\rho_{}}{\psi}. 
\end{align}

It is also possible to estimate matrix entries $\innerp{\phi}{\rho_{}}{\psi}$. If we know the state-preparation circuits $U\ket{0} = \ket{\psi}$ and $V\ket{0} = \ket{\phi}$ we can use the well-known Hadamard test as follows. We start from $\dyad{+} \otimes \rho_{}$ and use a circuit $C_{U,V}$ that applies $U^\dag$ if the ancilla is in state $\ket{1}$, and $V^\dag$ if it is in state $\ket{0}$. Finally, we measure $\sigma_X \otimes \dyad{0}$. It can be verified that
\begin{align}
    \tr[ (\sigma_X \otimes \dyad{0}) C_{U,V} (\dyad{+} \otimes \rho_{}) C_{U,V}^\dag ] = \Re\innerp{\phi}{\rho_{}}{\psi}. 
\end{align}
The imaginary part is obtained by measuring $\sigma_Y \otimes \dyad{0}$ instead.

Another interesting task is the estimation of the normalization constant, \ie\ the denominator in Eq.~\eqref{eq:expected_rho_app}. This quantity can be interpreted as the probability that the  algorithm stops without ever performing a restart (see Fig.~\ref{fig:algorithm_diagram} for intuition). This fact is used in Ref.~\cite[Proposition 11]{zhang2023dissipative} to estimate the partition function of Gibbs states. More generally, suppose we execute the algorithm $n_\mathrm{runs}$ times and keep track of the total number of restarts $n_\mathrm{restarts}$. Then
\begin{align}
\label{eq:norm_const_estimation}
    \sum_{k=0}^{T} r_k R_k \tr[\mathcal{E}^k(\rho_0)] = \sum_{k=0}^T \mathbb{P}(\text{stop at iteration } k \mid \text{algorithm not restarted}) \approx \frac{n_\mathrm{runs}}{n_\mathrm{runs} + n_\mathrm{restarts}} . 
\end{align}

As an example of application, consider the estimation of $\tr(X)$ where $X$ is the solution to the discrete-time Lyapunov equation $AXA^\dag - X +B=0$. From Theorem~\ref{thm:discrete_lyapunov} the expected output state of the algorithm is $\epsilon$-close in trace distance to the normalized solution. From Appendix~\ref{app:proof_error_discrete} we have that the outputs state is normalized by $\sum_{k=0}^{T} r_k R_k \tr[\mathcal{E}^k(\rho_0)] = \tr[\mathcal{S}(T)]$, and thus we can use Eq.~\eqref{eq:norm_const_estimation} as an estimator. Inspecting Eq.~\eqref{eq:upperbound_discrete} in the proof of the theorem we find an upper bound on the relative error of the normalization constant:
\begin{align}
    \frac{1}{2}\tracenorm*{ \frac{X}{\tr(X)} - \mathbb{E}[\rho] }
    \leq
    \underbrace{\frac{\tr[\mathcal{S}(\infty)] - \tr[\mathcal{S}(T)]}{\tr[\mathcal{S}(\infty)]}}_{
    \textstyle
    \left| \frac{\tr(X) - \tr[\mathcal{S}(T)] }{\tr(X)} \right| \; \text{(relative error)}
    } \leq
    \opnorm{A}^{2(T+1)}
    \leq
    \epsilon .
\end{align}
Thus the relative bias of the estimator decreases exponentially in $T$.

\section{Comparison of algorithms for the continuous-time Lyapunov equation}
\label{app:comparison_of_algos}

Here we discuss four quantum algorithms  for solving the continuous-time Lyapunov equation. The algorithms are highly heterogeneous with respect to the assumptions on the problem. To make the comparison possible we focus on the problem instance $AX + XA^\dagger + B = 0$ where $A, B \in \mathbb{C}^{N\times N}$, $A$ is negative definite with $\opnorm{A}=1/2$, and $B$ is positive semidefinite with $\tr[B]=1$.
Furthermore, the algorithms use different query access models and produce output in different formats. Thus their runtime/query complexity comparison is for informational purposes and should not be interpreted as conclusive.
Table~\ref{tab:comparison} summarizes the algorithms in terms of input, output, primitives and runtime or query complexity, where the $\tilde{\mathcal{O}}$ notation suppresses polylogarithmic factors and $\kappa$ denotes the condition number of $A$. We did not include~\cite[Lemma 5]{Liu_2025}, which can solve the Riccati equation $AX + X^\dag A^\dag +B = X^\dag D X$ when $D,B \succ 0$ and $D^{-1}A$ is Hermitian. To solve our problem instance we would need $D=0$, which is not allowed.

\begin{table}[ht]
\small
\centering
\begin{talltblr}[
    note{a} = {Includes amplitude amplification. Could be improved by replacing HHL with recent quantum linear system solvers.},
    note{b} = {Success probability $p$ not included. Boosting success probability to $1$ would lead to an additional factor $\mathcal{O}(1/p)$ from amplitude amplification.},
    note{c} = {For some real constant $\alpha \neq 0$.},
    entry=none,
    label=none
]{
    width = \linewidth,
    colspec = {l*{4}{X[c]}}, 
    hline{1,Z} = {0.8pt},    
    hline{2} = {0.4pt},      
    row{1} = {font=\bfseries},
}
    Algorithm
    & Input
    & Output
    & Primitives
    & Runtime or query complexity
    \\
    \cite{sunSolvingLyapunovEquation2017}
    & $s$-sparse matrix oracle for $A$, state preparation of $\ket{\text{vec}(-B)}$
    & Pure state $\ket{\tilde x}$ with $\lnorm{\ket{\tilde x} - \ket{\vecop(X)}} \leq \epsilon$
    & HHL
    & $\tilde{\mathcal{O}}\left(\frac{\log(N) s^2 \kappa^2}{\epsilon}\right)$\TblrNote{a}
    \\
    \cite[Theorem 4]{claytonDifferentiableQuantumComputing2024}
    & Block encoding of $A$, $B$
    & $(\opnorm{B}, \epsilon)$ block encoding of $X$
    & LCU, imaginary-time Hamiltonian simulation
    & $\tilde{\mathcal{O}}\left(\sqrt{\frac{\opnorm{B} \kappa^5}{\epsilon}} \right)$\TblrNote{b}
    \\
    \cite[Theorem 1]{Somma_2025} 
    & Block encoding of $A$, $B$
    & $(\alpha\kappa\opnorm{B}\log(\frac{\kappa\opnorm{B}}{\epsilon}), \epsilon)$ block encoding of $X$\TblrNote{c}
    & LCU, real-time Hamiltonian simulation
    & $\mathcal{O}\left(\kappa \polylog\left(\frac{\kappa \opnorm{B}}{\epsilon}\right)\right)$\TblrNote{b}
    \\
    This work
    & Block encoding of $A$, state preparation for $B$
    & Mixed state $\mathbb{E}[\rho]$ with $\tracenorm*{\mathbb{E}[\rho] - \frac{X}{\tr[X]}} \leq \epsilon$
    & Conditional restarts, imaginary-time Hamiltonian simulation
    & $\mathcal{O}\left(\frac{\kappa^3}{\epsilon} \log(\frac{1}{\epsilon}) \log(\frac{\kappa}{\epsilon}) \right)$
    \\
\end{talltblr}
\begingroup
\setlength{\abovecaptionskip}{6pt}
\caption{\textbf{Quantum algorithms for the continuous-time Lyapunov equation $AX + XA + B =0$.} For all algorithms we assume negative definite $A$ with $\opnorm{A}=1/2$ and $B$ positive semidefinite with $\tr[B]=1$. A $(\alpha, \epsilon)$ block encoding $U_X$ of $X$ means $\opnorm{X - \alpha\Pi U_X\Pi} \leq \epsilon$ for a suitable projector $\Pi$. Query complexity is in terms of access to the stated input oracles.}
\label{tab:comparison}
\endgroup
\end{table}

\subsection{Sun and Zhang\texorpdfstring{~\cite{sunSolvingLyapunovEquation2017}}{}}

\vspace{-5pt}

The algorithm in Ref.~\cite{sunSolvingLyapunovEquation2017} prepares a pure quantum state proportional to the solution of the vectorized Lyapunov equation $(A\otimes I + I\otimes A) \vecop(X) = \vecop(-B)$, where $I$ is the $N\times N$ identity and $\vecop(B)$ stacks the columns of $B$. Their input consists of a quantum circuit that prepares the state $\ket{\vecop(-B)} = \vecop(B)/\lnorm{\vecop(B)}$ and an oracle for $(A\otimes I + I\otimes A)$. Assuming $A$ is also $s$-sparse, the latter oracle can be implemented with a constant number of queries to a sparse oracle for $A$. Their work then applies the HHL algorithm, which is a quantum linear system solver, to approximate the normalized state $\ket{\vecop(X)}$. This state is a normalized solution to the equation up to error $\epsilon$ in $\ell_2$ distance. The runtime of HHL is $\tilde{\mathcal{O}}(\log(N) s^2\kappa^2/\epsilon)$, which, however, could be improved by using modern quantum linear system solvers.

\vspace{-15pt}

\subsection{Clayton et al.\texorpdfstring{~\cite{claytonDifferentiableQuantumComputing2024}}{}}

\vspace{-5pt}

The algorithm in Ref.~\cite{claytonDifferentiableQuantumComputing2024} constructs a block encoding of the solution $X$. It begins by discretizing the integral solution to the continuous-time Lyapunov equation into a sum. Recall from our main text that the $k$-th term of the sum is of the form $e^{k \Delta A} B e^{k \Delta A^\dag}$. In Ref.~\cite{claytonDifferentiableQuantumComputing2024} a block-encoding oracle for $B$ is provided by assumption, while block encodings for $e^{k \Delta A}$ and $e^{k \Delta A^\dag}$ shall be produced by imaginary-time Hamiltonian simulation methods. It is worthwhile to note that the product of block encodings is not the block encoding of the products in general, $U_{e^{k\Delta A}} U_B U_{e^{k \Delta A^\dag}} \neq U_{e^{k\Delta A} B e^{k \Delta A^\dag}}$, so achieving this will require additional ancilla registers~\cite[Lemma 53]{gilyenQuantumSingularValue2019}. Clayton et al. then take a linear combination of unitaries (LCUs) with all the terms in the sum. The final result is obtained upon postselecting all the ancilla registers on the $\ket{0}$ state, which introduces a probability success. We do not include this factor in the complexity in Table~\ref{tab:comparison}. Assuming the best possible normalization factor for the block encodings of $B$ and $A$, \ie\ their operator norm, the algorithm outputs an approximation to $X/\opnorm{B}$. For the query complexity of the algorithm we refer to the proof in Ref.~\cite[Theorem 4]{claytonDifferentiableQuantumComputing2024}, which gives $\tilde{\mathcal{O}}(\sqrt{\opnorm{B} \kappa^5 /\epsilon})$.

\vspace{-8pt}

\subsection{Somma et~al.\texorpdfstring{~\cite{Somma_2025}}{}}

\vspace{-5pt}

The algorithm in Ref.~\cite{Somma_2025} constructs a block encoding of the solution $X$ using an LCU and real-time Hamiltonian simulation as the main primitives. Their algorithm solves the more general problem $A'X + XB' = C'$. To specialize to our problem instance we set $A'=-A$, $B'=-A^\dagger=-A$, $C'=B$. 
Inputs of the algorithm are block encodings of the matrices $A$ and $B$. 
The output is a quantum circuit $U_{X/x}$ for a block encoding of the normalized solution matrix $X/x$ with $x\geq\kappa\opnorm{B}$, $Q=-A\otimes I - I\otimes A$. The vectorized form of the Lyapnuov equation reads $\vecop(X) = Q^{-1}\vecop(B)$. Now Somma~et~al. consider an integral representation of $Q^{-1}$ and express $X$ in terms of this integral after reverting the vectorization of the Lyapunov equation. Truncating and discretizing the integrals results in an LCU applied to the block encoding of $B$. The unitaries in the LCU are Hamiltonian simulations of the form $e^{-i\omega t A}$ with real $\omega$. Reference~\cite{Somma_2025} considers an implementation of the Hamiltonian simulation via QSVT to error $\opnorm{\Pi U_{X/x}\Pi - X/x} \leq \epsilon$. Note that $Q\succ 0$ and we assume access to a block encoding of $A$ for comparable input models. Hence, we use Theorem~1, Case~2~\cite{Somma_2025} to achieve query complexity $\mathcal{O}[\kappa\polylog(\kappa \opnorm{B}/ \epsilon)]$ for $x \propto \kappa\opnorm{B}\log(\kappa\opnorm{B}/\epsilon)$. 
To use the block encoding $U_{X/x}$ in a concrete task, such as preparing a state $\propto X\ket{\psi}$ or computing matrix elements $\innerp{\phi}{X}{\psi}$, requires additional queries to the oracle owing to the success probability $p<1$ of preparing the block encoding. This leads to an additional factor of either $\mathcal{O}(1/\sqrt{p})$ if amplitude amplification is available or $\mathcal{O}(1/p)$ if the state preparation is repeated upon the wrong postselection outcome. We do not include this factor in the complexity in Table~\ref{tab:comparison}.
Note that their Theorem~1, Case~4~\cite{Somma_2025} for $Q\succ 0$ requires access to a block encoding of $\sqrt{A}$ but gives improved query complexity. This assumption would also improve the query complexity of our algorithm but we do not consider it here.

\vspace{-15pt}

\subsection{Our work}

\vspace{-5pt}

Consider our algorithm for the continuous-time Lyapunov equation. Recall that we have three additive error terms in our algorithm: discretization of integral $\epsilon_1$, finite limit of integration $\epsilon_2$, and approximate implementation error $\epsilon_{\sim}$. Suppose we want a total error $\epsilon$, so we set each individual term to $\epsilon/3$. This requires a block encoding of $M=e^{\Delta A}$. As discussed in Appendix~\ref{app:block_encoding_exp}, given access to a block encoding of $-A$ we can construct a suitable approximate block encoding with $\mathcal{O}\left[\kappa \log (T/\epsilon)\right]$ queries at each step. From Theorem~\ref{thm:continuous_lyapunov} our algorithm uses parameters $T = \frac{3\kappa^2}{2\epsilon} \log(\frac{3}{\epsilon})$ and $\Delta = \frac{\epsilon}{3\kappa}$. 
By Corollary~\ref{cor:stopping_time} our algorithm is expected to run for $T+1$ steps. Thus we expect to require at most
\begin{equation}
    \mathcal{O}\left[T\kappa\log(\frac{T}{\epsilon})\right]
    =
    \mathcal{O}\left[\frac{\kappa^3}{\epsilon} \log(\frac{1}{\epsilon}) \log(\frac{\kappa}{\epsilon}) \right]
\end{equation}
queries to the block encoding of $-A$. Here we dropped $\log \log(1/\epsilon)$ terms.

\twocolumngrid
\bibliography{main}

\end{document}